\pdfoutput=1
\documentclass[lettersize,journal]{IEEEtran}

\usepackage{amsmath,amsfonts,amssymb,amsthm}
\usepackage{bm}

\usepackage{algorithm}
\usepackage{algpseudocode}

\usepackage{array}
\usepackage{booktabs}
\usepackage{multirow}

\usepackage{graphicx}
\usepackage{caption}
\usepackage{subfig}  

\usepackage{cite}
\usepackage{url}
\usepackage{hyperref}  

\usepackage{color}
\usepackage{verbatim}
\usepackage{textcomp}
\usepackage{stfloats}
\usepackage{gensymb}

\newcommand{\diag}{\mathop{\mathrm{diag}}}  

\newtheorem{proposition}{Proposition}
\newtheorem{theorem}{Theorem}

\def\ie{$i.e.$}

\def\red#1{\textcolor{red}{#1}}

\begin{document}

\title{PointNCBW: Towards Dataset Ownership Verification for Point Clouds via Negative Clean-label Backdoor Watermark}

\author{
Cheng Wei, Yang Wang, Kuofeng Gao, Shuo Shao, Yiming Li, Zhibo Wang, and Zhan Qin

\thanks{The first two authors contributed equally to this paper.}
\thanks{Corresponding Author(s): Yiming Li and Zhan Qin.}
\thanks{
Cheng Wei, Yang Wang, Shuo Shao, Zhibo Wang, and Zhan Qin are with the State Key Laboratory of Blockchain and Data Security, Zhejiang University, Hangzhou 310007, China and also with Hangzhou High-Tech Zone (Binjiang) Institute of Blockchain and Data Security, Hangzhou 310053, China. (e-mail: \{weiccc, Kjchz, shaoshuo\_ss, zhibowang, qinzhan\}@zju.edu.cn)}
\thanks{
Yiming Li is now with College of Computing and Data Science, Nanyang Technological University, Singapore 639798. He was with the State Key Laboratory of Blockchain and Data Security, Zhejiang University, Hangzhou 310007, China. (e-mail: liyiming.tech@gmail.com)}
\thanks{Kuofeng Gao is with Tsinghua Shenzhen International Graduate School, Tsinghua University, Shenzhen 518055, China.
(e-mail: gkf21@mails.tsinghua.edu.cn).}}



\markboth{IEEE Transactions on Information Forensics and Security}%
{IEEE Transactions on Information Forensics and Security}


\maketitle

\begin{abstract}
Recently, point clouds have been widely used in computer vision, whereas their collection is time-consuming and expensive. As such, point cloud datasets are the valuable intellectual property of their owners and deserve protection. To detect and prevent unauthorized use of these datasets, especially for commercial or open-sourced ones that cannot be sold again or used commercially without permission, we intend to identify whether a suspicious third-party model is trained on our protected dataset under the black-box setting. We achieve this goal by designing a \emph{scalable} clean-label backdoor-based dataset watermark for point clouds that ensures both effectiveness and stealthiness. Unlike existing clean-label watermark schemes, which were susceptible to the number of categories, our method can watermark samples from all classes instead of only from the target one. Accordingly, it can still preserve high effectiveness even on large-scale datasets with many classes. Specifically, we perturb selected point clouds with non-target categories in both shape-wise and point-wise manners before inserting trigger patterns without changing their labels. The features of perturbed samples are similar to those of benign samples from the target class. As such, models trained on the watermarked dataset will have a distinctive yet stealthy backdoor behavior, $i.e.$, misclassifying samples from the target class whenever triggers appear, since the trained DNNs will treat the inserted trigger pattern as a signal to deny predicting the target label. We also design a hypothesis-test-guided dataset ownership verification based on the proposed watermark. Extensive experiments on benchmark datasets are conducted, verifying the effectiveness of our method and its resistance to potential removal methods. The codes of our method are available at \href{https://github.com/weic0810/PointNCBW}{GitHub}.
\end{abstract}

\begin{IEEEkeywords}
Dataset Ownership Verification, Backdoor Watermark, Dataset Copyright Protection, 3D Point Clouds.
\end{IEEEkeywords}

\section{Introduction}
\label{sec:intro}
Point clouds have been widely and successfully adopted in many vital applications ($e.g.$, autonomous driving \cite{chen20203d} and augmented reality \cite{lim2022point}) since they can provide rich geometric, shape, and scale information \cite{fan2023mba}. In particular, collecting point clouds is even more time-consuming and expensive compared to classical data types ($e.g.$, image or video). It necessitates using costly 3D sensors and intricate data processing procedures ($i.e.$, registration \cite{lv2023kss} and filtering \cite{zhang2020pointfilter}). As such, point cloud datasets are valuable intellectual property.


Due to the widespread applications, point cloud datasets are likely to be publicly released as open-sourced or commercial datasets. However, to the best of our knowledge, almost all existing methods cannot be directly exploited to protect their copyright when they are publicly released. In such scenarios, the adversaries may train their commercial models on open-sourced datasets that are restricted to academic or research purposes or even on commercial ones that have been illegally redistributed. Arguably, the protection difficulty stems mainly from the publicity of these datasets and the black-box nature of the suspicious models, since existing traditional data protection schemes either hinder the dataset accessibility ($e.g.$, data encryption \cite{acar2018survey}), requir manipulation of model training ($e.g.$, differential privacy \cite{abadi2016deep}), or demand accessing training samples during the verification process ($e.g.$, digital watermark \cite{ferreira2020robust}).

\begin{figure}[!t]
\vspace{-1em}
	\centering  
  \captionsetup[subfloat]{font=small}
	\subfloat[{not stealthy}]{
		\includegraphics[scale=0.45]{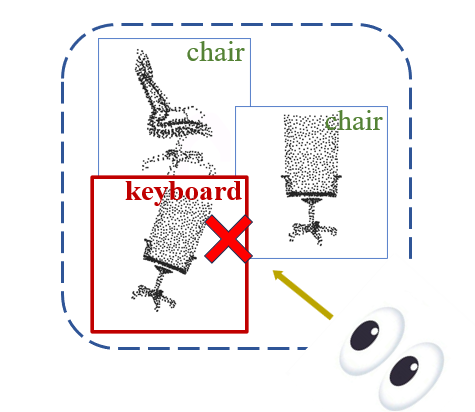}}
     \hfill
	\subfloat[non-scalable]{
		\includegraphics[scale=0.22]{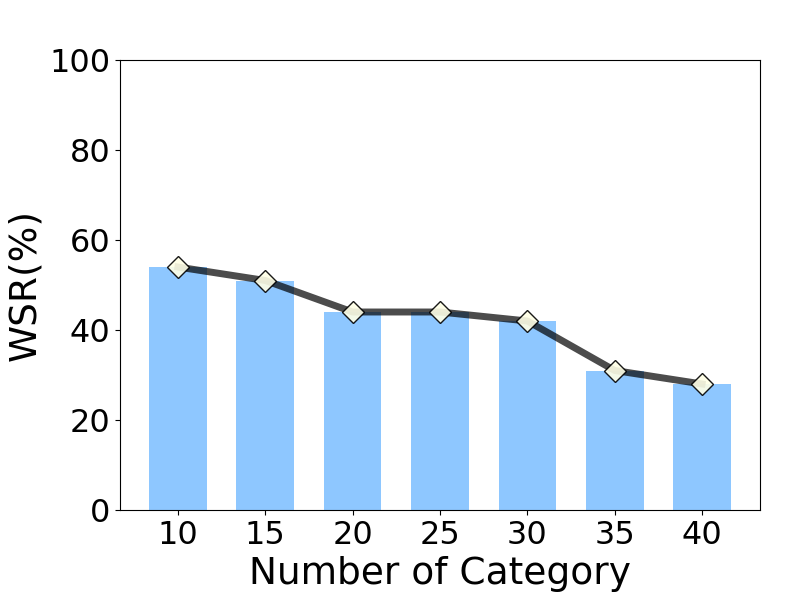}}
	\caption{The limitations of existing backdoor attacks that could be used as watermarks to protect point cloud datasets. \textbf{(a)} Existing poison-label backdoor watermarks ($i.e.$, PCBA \cite{xiang2021backdoor}, PointPBA \cite{li2021pointba}, IRBA \cite{gao2023imperceptible}) are not stealthy under human inspection due to sample-label mismatch. \textbf{(b)} The only existing clean-label backdoor watermark ($i.e.$, PointCBA \cite{li2021pointba}) has limited effect (measured by watermark success rate (WSR)) when the protected dataset contains many categories.}
 \label{fig:limitations}
 \vspace{-1.3em}
\end{figure}

Currently, to the best of our knowledge, dataset ownership verification (DOV) \cite{li2023black,li2022untargeted,tang2023did,guo2023domain} is the only promising approach that can be generalized to protect the copyright of public point cloud datasets. This method was initially and primarily used to safeguard the copyright of image datasets. Specifically, in existing DOV methods, dataset owners adopted and designed backdoor attacks \cite{li2022backdoor} to watermark the original dataset before releasing it. Given a suspicious third-party black-box model that can only be accessed via API, dataset owners can verify whether it is trained on the protected dataset by inspecting whether it has owner-specified backdoor behaviors. The backdoor-based DOV successfully safeguards publicly accessible datasets, particularly in the presence of black-box suspicious models. As such, the key to protecting the copyright of point cloud datasets via DOV lies in designing a suitable backdoor watermark for them.


However, directly applying existing backdoor attacks against 3D point clouds as watermarks faces several challenges, as shown in Figure \ref{fig:limitations}. Firstly, except for PointCBA \cite{li2021pointba}, all existing point cloud backdoor methods ($i.e.$, PCBA \cite{xiang2021backdoor}, PointPBA \cite{li2021pointba}, IRBA \cite{gao2023imperceptible}) used poisoned labels, in which the assigned labels of watermarked samples are different from their ground-truth ones. Accordingly, these watermarks lack \emph{stealthiness} as they can be easily detected and removed by malicious dataset users who scrutinize the correlation between point clouds and their corresponding labels. Secondly, the performance of the only existing clean-label backdoor watermark, \ie, PointCBA, is not \emph{scalable}. In other words, its watermark performance will significantly decrease when datasets contain many categories, preventing its application as a watermark for large-scale point cloud datasets.

We find that the non-scalability of existing clean-label backdoor watermarks for both images and point clouds \cite{turner2019label,li2021pointba,gao2023not} comes from their common poisoning paradigm, where defenders can only watermark samples from a specific category ($i.e.$, target class). As such, the more categories of samples in the dataset, the smaller the maximum watermarking rate ($i.e.$, maximal proportion of samples for watermark), resulting in a reduction in overall watermark effects. We argue that this defect is primarily caused by the \emph{positive} trigger effects of existing clean-label watermarks. Specifically, these methods have to add triggers to samples from the target class, aiming to build a positive connection between the trigger pattern and the target label ($i.e.$, adding triggers to any benign sample increases the probability of being predicted as the target label). 

In this paper, motivated by the aforementioned understandings, we propose to design a scalable clean-label backdoor watermark by introducing the \emph{negative} trigger effects, where we intend to decrease the prediction confidence of watermarked models to samples from the target class when owner-specified trigger patterns arise. Specifically, before implanting trigger patterns, we first perturb selected point clouds from non-target categories so that their features lie close to those from the target class. After that, we implant the trigger patterns into these samples to generate the watermarked data. The labels of these watermarked samples are different from the target label, but they lie close to those of the target class in feature space. Accordingly, the trained DNNs will learn to treat the inserted trigger pattern as a signal to deny predicting the target label. This method is called \textbf{n}egative \textbf{c}lean-label \textbf{b}ackdoor \textbf{w}atermark for \textbf{point} clouds (\textbf{PointNCBW}). Our proposed PointNCBW surpasses existing backdoor-based DOV methods in two aspects. Firstly, the watermarks generated by PointNCBW are inherently more resistant to scrutiny by malicious dataset users due to the consistency between point clouds and their corresponding labels. This alignment significantly enhances the stealthiness of the watermarks. Secondly, negative triggers facilitate ownership verification based on PointNCBW to work effectively in large-scale datasets with numerous categories, leading to superior scalability. We further substantiate this scalability claim through comparative experimentation, detailed in Section \ref{sec Scalability}. In addition, we also design a hypothesis-test-guided dataset ownership verification based on our PointNCBW by examining whether the suspicious model shows less confidence in point clouds containing owner-specified triggers from the target class. It alleviates the adverse effects of randomness introduced by sample selection.

The main contributions of this paper are four-fold: 
\begin{itemize}
\item{We explore how to protect the copyright of point cloud datasets via dataset ownership verification (DOV).}
\item{We reveal the limitations of using existing backdoor attacks against point clouds for DOV and their reasons.}
\item{We propose the first scalable clean-label backdoor watermark for point cloud datasets ($i.e.$, PointNCBW).}
\item{We conduct extensive experiments on benchmark datasets, which verify the effectiveness of our PointNCBW and its resistance to potential adaptive attacks.}
\end{itemize}

\section{Related Work}

\subsection{Deep Learning on 3D Point Clouds}
With the advent of deep learning, point-based models have gained significant popularity owing to their exceptional performance in various 3D computer vision tasks. Qi $et\ al.$ \cite{pointnet} first proposed PointNet, which directly processes raw point cloud data without needing data transformation or voxelization. It adopts the symmetric function, max-pooling, to preserve the order-invariant property of point clouds. To learn the local features of point clouds, they further proposed a hierarchical network PointNet++ \cite{pointnet++}, which can capture geometric structures from the neighborhood of each point better. Inspired by them, subsequent research \cite{dgcnn,pointcnn,yan2020pointasnl,zhao2021point} in point cloud-based deep learning has emerged following similar principles. 

\subsection{Backdoor Attack}

Backdoor attack is an emerging yet critical training-phase threat to deep neural network (DNNs) \cite{li2022backdoor}. In general, the adversaries intend to implant a latent connection between the adversary-specified trigger patterns and the malicious prediction behaviors (\ie, backdoor) during the training process. The attacked models behave normally on benign samples, whereas their prediction will be maliciously changed whenever the trigger pattern arises. Currently, most of the existing backdoor attacks were designed for image classification tasks \cite{gu2019badnets,li2021invisible,qi2023revisiting,gao2023sbackdoor,bai2024badclip,yang2024not}, although there were also a few for others \cite{zhai2021backdoor,li2022few,gao2023backdoor,xiang2024badchain,fan2024stealthy,cai2024towards}. These methods could be divided into two main categories: \emph{poison-label} and \emph{clean-label} attacks, depending on whether the labels of modified poisoned samples are consistent with their ground-truth ones \cite{turner2018clean}. In particular, clean-label attacks are significantly more stealthy than poison-label ones since dataset users cannot identify them based on the image-label relationship even when they can catch some poisoned samples. However, as demonstrated in \cite{gao2023not}, clean-label attacks are significantly more challenging to succeed due to the antagonistic effects of `robust features' related to the target class contained in poisoned samples.


Currently, a few studies have focused on 3D point clouds \cite{liu2019extending,he2023generating}, particularly in the context of backdoor attacks \cite{xiang2021backdoor,li2021pointba,gao2023imperceptible}. Although these attacks also
targeted the classification tasks, developing effective backdoor attacks for 3D point clouds is challenging due to inherent differences in data structure and deep learning model architectures~\cite{li2021pointba}.
As we will show in our experiments, the only existing clean-label attack (\ie, PointCBA \cite{li2021pointba}) exhibits significantly reduced efficacy compared to its poison-label counterparts (\ie, PCBA \cite{xiang2021backdoor}, PointPBA-Ball, PointPBA-Rotation \cite{li2021pointba}, and IRBA \cite{gao2023imperceptible}). This difference is especially noticeable when the dataset has many categories, revealing a lack of scalability for large datasets. Therefore, designing effective and scalable clean-label backdoor attacks against 3D point clouds to address this issue is still an ongoing challenge that needs more research.

\subsection{Dataset Protection}
Dataset protection is a longstanding research problem, aiming to prevent the unauthorized use of datasets. Existing classical dataset protection methods involve three main categories: data encryption, differential privacy, and digital watermarking. Especially, data encryption \cite{acar2018survey} encrypted the protected datasets so that only authorized users who hold a secret key for decryption can use it; differential privacy \cite{abadi2016deep} prevented the leakage of sensitive personal information during model training; digital watermark \cite{ferreira2020robust} embeded an owner-specified pattern to the protected data for post-hoc examination. However, these methods are unable to protect publicly released datasets ($e.g.$, ImageNet) from being used to train third-party commercial models without authorization due to the public availability of datasets and the black-box accessing nature of commercial models \cite{li2023black}. Recently, unlearnable examples \cite{huang2021unlearnable} were also proposed to protect the datasets by directly preventing them from being learned by DNNs. However, this method cannot be used to protect public datasets since it usually requires the modification of all samples and compromises dataset utilities.

Dataset ownership verification (DOV) \cite{li2023black,li2022untargeted,tang2023did,guo2023domain,zhao2024zero} intends to verify whether a given suspicious model is trained on the protected dataset under the black-box setting, where defenders can only query the suspicious model. To the best of our knowledge, this is currently the only feasible method to protect the copyright of public datasets. Specifically, existing DOV methods intend to implant specific (backdoor) behaviors in models trained on the protected dataset while not reducing their performance on benign samples. Dataset owners can verify ownership by examining whether the suspicious model has specific backdoor behaviors. However, existing DOV methods are mostly designed for image classification datasets. How to protect other types of datasets is left far behind.

\begin{figure*}[!t]
    \centering 
    \includegraphics[width=0.95\linewidth]{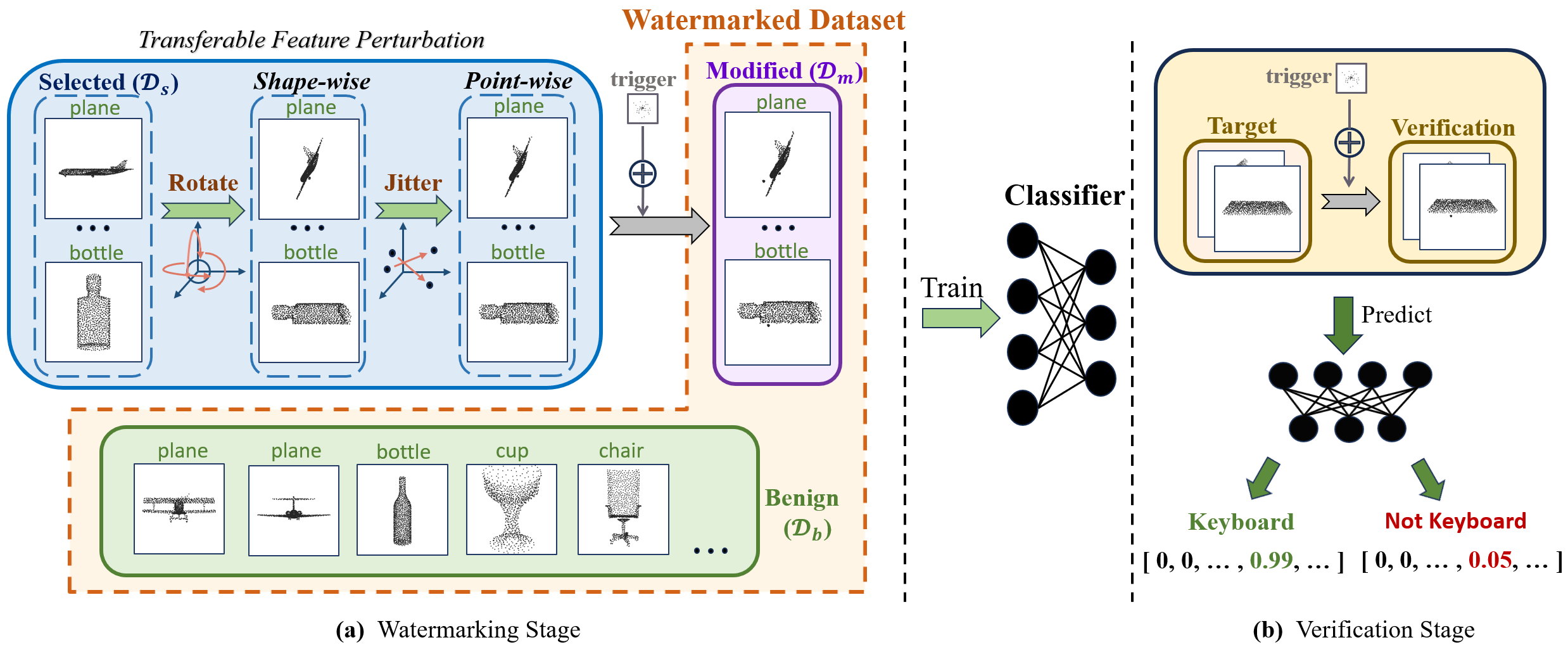}
    \caption{The main pipeline of dataset ownership verification based on our negative clean-label backdoor watermark for point clouds (PointNCBW). In the \textbf{watermarking stage}, we generate the watermarked version of the original dataset. Specifically, we design and exploit transferable feature perturbation (TFP) to perturb a few selected point clouds from the original dataset with non-target categories so that they lie close to those from the target class in the feature space. Our TFP has two steps, including shape-wise and point-wise perturbations, to ensure transferability across model structures. In the \textbf{verification stage}, we verify whether a suspicious third-party model is trained on our protected dataset by examining whether it misclassifies samples from the target class containing owner-specified trigger patterns via the hypothesis test.}
    \label{fig:main}
    \vspace{-0.8em}
\end{figure*}

\section{Negative Clean-label Backdoor Watermark for Point Clouds (PointNCBW)}
\label{A Novel Point Cloud Backdoor Scheme}

In this section, we present a comprehensive description of our proposed method (\ie, PointNCBW). Before we illustrate its technical details, we first describe our threat model and briefly review the main pipeline of backdoor watermarks.

\subsection{Preliminaries}

\noindent \textbf{Threat Model.} 
This paper focuses on backdoor watermarks for point cloud datasets in classification tasks. Specifically, the dataset owner can watermark some benign point clouds before releasing the victim dataset. Dataset users will exploit the released (watermarked) datasets to train and deploy their DNNs but will keep their training details private. Accordingly, dataset owners can only examine whether a suspicious model is trained on their watermarked point cloud dataset by its prediction behaviors under the \emph{black-box} access setting.

\vspace{0.3em}

\noindent \textbf{The Main Pipeline of Existing Backdoor Watermarks for Point Clouds.} 
Let $\mathcal{D}=\{(\bm{x}_i,y_i)\}_{i=1}^N$ denotes the benign dataset containing $N$ point clouds. Each point cloud $\bm{x}_i$ contains $M_i$ points ($i.e.$, $\bm{x}_i \in \mathbb{R}^{3 \times M_i}$) whose label $y_i \in \{1, 2, \cdots, K\}$. How to generate the watermarked dataset $\mathcal{D}_w$ is the cornerstone of all backdoor watermarks. Currently, all backdoor watermarks for point clouds are targeted and with \emph{positive} trigger effects. In other words, adding triggers increases the probability that watermarked DNNs predict samples to the target class $y^{(t)}$. Specifically, $\mathcal{D}_w$ has two disjoint parts, including the modified version of a selected subset ($i.e.$, $\mathcal{D}_s$) of $\mathcal{D}$ and remaining benign point clouds, $i.e.$, $\mathcal{D}_w=\mathcal{D}_m \cup \mathcal{D}_b$, where $\mathcal{D}_b=\mathcal{D} \backslash \mathcal{D}_s$, $\mathcal{D}_m=\{(\bm{x}', y^{(t)})|\bm{x}'=G(\bm{x}),(\bm{x}, y)\in \mathcal{D}_s\}$, $G: \mathbb{R}^{3 \times M} \rightarrow \mathbb{R}^{3 \times M}$ is the owner-specified generator of watermarked samples. $\lambda \triangleq \frac{|\mathcal{D}_m|}{|\mathcal{D}|}$ is the \emph{watermarking rate}. In general, backdoor watermarks are mainly characterized by their watermark generator $G$. For example, $G(\bm{x}) = (\bm{I}-\diag{\bm{\delta}}) \cdot \bm{x} + \diag{\bm{\delta}} \cdot \bm{\Gamma}$, where $\bm{\delta}$ is a 0-1 vector, $\bm{I}$ is the identity matrix, and $\bm{\Gamma}$ is pre-defined trigger pattern in PointPBA-Ball \cite{li2021pointba}. In particular, in existing clean-label backdoor watermarks ($e.g.$, PointCBA \cite{li2021pointba}), dataset owners can only watermark samples from the target class, $i.e.$, $\mathcal{D}_s \subset \mathcal{D}^{(t)} \triangleq \{(\bm{x}, y)| (\bm{x}, y) \in \mathcal{D}, y = y^{(t)} \}$. As such, their watermarking rate is $\frac{1}{K}$ at most for class-balanced datasets. This limits their performance when the number of categories in the victim dataset is relatively large.

\subsection{The Overview of PointNCBW}
\label{Watermark Method for Point Cloud Dataset}


In this paper, we design a clean-label backdoor watermark for point clouds with \emph{negative} trigger effects to overcome the limitations of existing backdoor watermarks (as demonstrated in our introduction). We denote our watermarking method as \textbf{n}egative \textbf{c}lean-label \textbf{b}ackdoor \textbf{w}atermark for \textbf{point} clouds (\textbf{PointNCBW}). 
In general, our PointNCBW consists of two main stages: transferable feature perturbation (TFP) and trigger implanting. Specifically, TFP perturbs selected point clouds with non-target categories so that they lie close to those from the target class in the hidden space defined by a pre-trained model. After that, we insert trigger pattern $\bm{\Gamma}$ to obtain modified point clouds $\mathcal{D}_m$ via
\begin{equation}
\label{our Dm}
    \mathcal{D}_m=\{(\bm{x}',y)|\bm{x}'=U(p(\bm{x}),\bm{\Gamma}),(\bm{x},y)\in \mathcal{D}_s\},
\end{equation}
where $U(p(\bm{x}),\bm{\Gamma})$ is our watermark generator, $p$ represents our TFP, and $U$ is our trigger implanting function implemented with random replacing function.

Since we don't change the label of these watermark samples, the watermarked DNNs will interpret inserted triggers as signals to deny predicting the target label. The main pipeline of our method is shown in Figure \ref{fig:main}.
We also provide a detailed explanation of the underlying reasons behind the effectiveness of PointNCBW through experimental analysis in Section \ref{sec: reason_explain}.

\vspace{-0.6em}
\subsection{Transferable Feature Perturbation (TFP)}
\vspace{-0.5em}
\label{Feature Perturbation}


\vspace{0.3em}
\noindent
\textbf{General Perturbation Objective.}
After selecting the target category $y^{(t)}$ and the source sample group $\mathcal{D}_s$, our objective is to perturb each sample in $\mathcal{D}_s$ to bring them closer to category $y^{(t)}$ in feature space. Specifically, we randomly select some samples from category $y^{(t)}$ denoted as $\mathcal{D}_t$, and utilize the features of $\mathcal{D}_t$ as an alternative to the features of category $y^{(t)}$.
Let $\bm{x}_s$ represent one sample in $\mathcal{D}_s$, our general objective of perturbation is formulated by

\begin{equation}
\label{equ:perturbation general obj}
    \begin{aligned}
     &\min_{p}{\frac{1}{|\mathcal{D}_t|}\sum_{\bm{x}_t\in \mathcal{D}_t}\mathcal{E}(g_f(p(\bm{x}_s)),g_f(\bm{x}_t))},
    \end{aligned}
\end{equation}
where $\mathcal{E}$ is a Eluer distance in feature space $\mathbb{R}^d$ and $g_f$ is the feature extracting function of point cloud. In practice, we implement the $g_f$ with the second-to-last layer output of our surrogate model $g$ for approximation.

However, since we employ a surrogate model for feature extraction, it is crucial to ensure that our feature perturbation remains effective under different networks beyond the surrogate one. This raises the concern of transferability, which refers to the ability of our watermark to work effectively across different model structures.
To enhance the transferability, we optimize general objective function in Eq.~(\ref{equ:perturbation general obj}) through transferable feature perturbation (TFP). The transferability is also empirically verified in Section \ref{Transferability}. Specifically, our TFP consists of two sequential steps, including \emph{shape-wise} and \emph{point-wise} perturbations, as outlined below. 


\vspace{0.3em}
\noindent
\textbf{Shape-wise Perturbation.}
\label{Shape-wise Perturbation}
Rotation is a common transformation of 3D objects. It has been empirically proven to be an effective and inconspicuous method to conduct point cloud perturbation with transferability in previous works since DNNs for point cloud are sensitive to geometric transformations \cite{zhao2020isometry, fan2022careful}. Arguably, the observed transferability through rotation may primarily result from the commonalities in point cloud feature processing employed by current DNNs designed for point clouds. Specifically, the widespread application of max-pooling operations \cite{pointnet,pointnet++,dgcnn,yan2020pointasnl} across these architectures potentially contributes to this phenomenon. While max-pooling is effective at capturing overall features in unordered point clouds, it has difficulty maintaining rotational invariance. The wide application of max-pooling in current DNNs used for point cloud not only makes the feature spaces of these models similar, but also helps our perturbations transfer better. The changes brought about by rotations are often misinterpreted by different models, which supports the effectiveness of our method. As such, we exploit it to design our shape-wise perturbation with transformation matrix $S$ defined as follows:
\begin{equation}
\label{matrix M}
    S(\theta) = R(\psi)\cdot R(\phi) \cdot R(\gamma),
\end{equation}
where $R(\psi),R(\phi)$, and $R(\gamma)$ are rotation matrix with Eluer angles $\psi,\phi,\gamma$. Finally, we have the objective function of shape-wise perturbation, as follows:
\begin{equation}
\label{rot loss}
     \mathcal{L}_s(\theta)=\frac{1}{|\mathcal{D}_t|}\sum_{\bm{x}_t\in \mathcal{D}_t}\mathcal{E}(g_f(\bm{x}_s\cdot S(\theta)),g_f(\bm{x}_t)).\\
\end{equation}

Specifically, we employ the gradient descent to minimize the loss function defined in Eq. (\ref{rot loss}). Besides, to alleviate the impact of local minima, we employ a strategy of random point sampling and choose the optimal starting point for the optimization process. We summarize the shape-wise feature perturbation in Algorithm \ref{alg:sfe}.

\begin{figure}[!t]
\vspace{-1em}
    \begin{algorithm}[H]
    \caption{Shape-wise feature perturbation.}\label{alg:sfe}
    \begin{algorithmic}[1]
    \Require{point cloud $\bm{x}$, rotation matrix $S(\theta)$, objective function $\mathcal{L}_s$, number of iterations $T$, step size $\beta$, number of starting points sampling $n$}
    \Ensure{new point cloud $\bm{x}'$}
    \State randomly sample $\{\theta_i\}_{i=1}^{n}$ from range $(0,2\pi)^3$
    \State $\theta_0 \gets \mathop{\arg\min \mathcal{L}_s(\theta_i)}$
    \For{$t=1$ to $T$}
    \State obtain gradient $g_t=\nabla_{\theta}{\mathcal{L}_s(\theta_{t-1})}$ through chain rule
    \State $\theta_t \gets \theta_{t-1} - \beta \cdot g_t$
    \EndFor
    \State $\bm{x}' \gets \bm{x} \cdot S(\theta_T)$
    \end{algorithmic}
    \end{algorithm}
    \vspace{-1.8em}
\end{figure}

\vspace{0.3em}
\noindent
\textbf{Point-wise Perturbation.}
\label{Point-wise Perturbation}
After the shape-wise perturbation where we can obtain the optimal solution $\hat{\theta}$ for Eq. (\ref{rot loss}), we jitter the point cloud sample on the point-wise level. We denote the offset of coordinates as $\Delta \bm{x}$, then the objective function of optimization of point-wise perturbation is 

\begin{equation}
\begin{aligned}
    \label{equ: mig loss}
     \mathcal{L}_p(\Delta \bm{x})=
     \frac{1}{|\mathcal{D}_t|}\cdot
     \sum_{\bm{x}_t\in \mathcal{D}_t} & {\mathcal{E}(g_f(\bm{x}_s\cdot S(\hat{\theta})+\Delta \bm{x}),g_f(\bm{x}_t))} \\ &+\eta \cdot \Vert\Delta \bm{x}\Vert_2, 
\end{aligned}
\end{equation}
where we incorporate a regularization term $\Vert\Delta \bm{x}\Vert_2$ into the loss function to limit the magnitude of perturbation in $\ell_2$ norm. To minimize the loss $\mathcal{L}_p$ and mitigate the problem of perturbed samples `overfitting' to the surrogate model, we incorporate momentum \cite{dong2018boosting} into the iterative optimization process to stabilize updating directions and escape from `overfitting' local minima. Specifically, we compute the momentum of gradients and perturb the coordinates of points in the opposite direction of the momentum. The optimization process of point-wise feature perturbation is summarized in Algorithm \ref{alg:pfe}.

\begin{figure*}[!t]
    \centering
    \vspace{-0.5em}
    \includegraphics[width=0.92\linewidth]{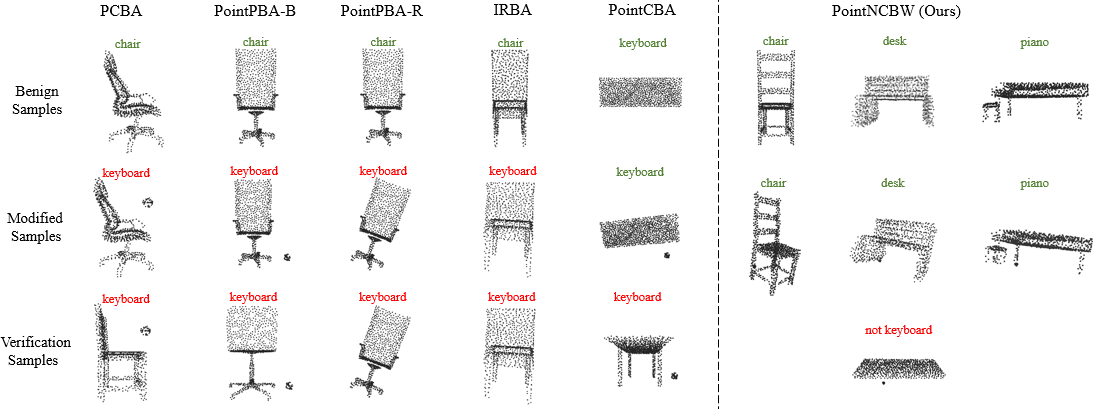}
    \caption{The example of point clouds involved in different backdoor watermarks.}
    \label{fig:samples}
\end{figure*}


\section{Ownership Verification via PointNCBW}
\label{Copyright Verification}
Given a suspicious black-box model $f(\cdot)$, we can verify whether it was trained on our protected point cloud dataset by analyzing its predictions on samples from the target class and their watermarked versions. To alleviate the side effects of randomness, we design a hypothesis-test-guided method, inspired by existing backdoor-based dataset ownership verification for images \cite{shao2024explanation,li2022untargeted,guo2023domain}, as follows. 

\begin{algorithm}[!t]
\caption{Point-wise feature perturbation.}\label{alg:pfe}
\begin{algorithmic}[1]
\Require{point cloud $\bm{x}$, objective function $\mathcal{L}_p$, number of iterations $T$, step size $\beta$, decay factor $\mu$}

\Ensure{new point cloud $\bm{x}'$}

\State $\Delta \bm{x}^{0}=0$, $g_0 = 0$
\For{$t=1$ to $T$}
\State obtain the gradient $g'=\nabla_{\Delta \bm{x}^{t-1}}{\mathcal{L}_p(\Delta \bm{x}^{t-1})}$
\State $g_t \gets g_{t-1} + \mu \cdot \frac{g'}{\Vert g' \Vert_2}$
\State $\Delta \bm{x}^{t} \gets \Delta \bm{x}^{t-1} - \beta \cdot g_t$
\EndFor
\State $\bm{x}'=\bm{x}+\Delta \bm{x}^{T}$
\end{algorithmic}
\end{algorithm}

\begin{proposition}
Suppose $f(\bm{x})$ is the posterior probability of $\bm{x}$ predicted by the suspicious model. Let variable $X$ denotes the benign sample from the target class $y^{(t)}$ and variable $X'$ is its verified version ($i.e.$ $X'=U(X,\Gamma)$). Let variable $P_b = f(X)_{y^{(t)}}$ and $P_v = f(X')_{y^{(t)}}$ denote the predicted probability of $X$ and $X'$ on $y^{(t)}$. Given null hypothesis $H_0: P_b=P_v+\tau$ ($H_1:P_b>P_v+\tau$), where hyper-parameter $\tau\in(0,1)$, we claim that the suspicious model is trained on the watermarked dataset (with $\tau$-certainty) if and only if $H_0$ is rejected.
\end{proposition}

In practice, we randomly sample $m$ different benign samples to conduct the pairwise $t$-test \cite{larsen2005introduction} and calculate its p-value. Additionally, we also calculate the confidence score $\Delta P=P_b-P_v$ to denote the verification confidence. The larger the $\Delta P$, the more confident the verification.

In particular, we can also prove that our dataset ownership verification can succeed if its watermark success rate is sufficiently large (which could be significantly lower than 100\%) when sufficient verification samples exist, as shown in Theorem \ref{them1}. Its proof is in our appendix.

\begin{theorem}\label{them1}
Suppose $f$ generates the posterior probability by a suspicious model. Let variable $X$ denote the benign sample from the target class $y^{(t)}$ and variable $X'$ is its verified version. Let $P_b = f(X)_{y^{(t)}}$ and $P_v = f(X')_{y^{(t)}}$ denote the predicted probability of $X$ and $X'$ on $y^{(t)}$. Assume that $P_b > \zeta$, we claim that dataset owners can reject the null hypothesis $H_0$ at the significance level $\alpha$, if the watermark success rate $W$ of $f$ (with $m$ verification samples) satisfies:

\begin{equation}
    \sqrt{m-1} \cdot (W + \zeta - \tau - 1) - t_{\alpha} \cdot \sqrt{W-W^2} > 0,
\end{equation}
where $t_{\alpha}$ is $\alpha$-quantile of t-distribution with $(m-1)$ degrees of freedom.
\end{theorem}

\begin{table*}[!t]
  \centering
  \caption{The watermark performance of different backdoor watermarks for point clouds on the ModelNet40 dataset (having 40 categories), the ShapeNetPart dataset (having 16 categories), and the PartNet dataset (having 24 categories), respectively. In particular, we mark the worst cases ($i.e.$, WSR $< 80\%$ or IoM $> 80\%$) in \red{red}.}
  \vspace{-0.3em}
  \scalebox{1.1}{
  \begin{tabular}{c|c|c|ccc|ccc}
    \toprule
     \multirow{2}{*}{Dataset$\downarrow$}&\multirow{2}{*}{Type$\downarrow$} & Model$\xrightarrow{}$& \multicolumn{3}{c|}{PointNet} & \multicolumn{3}{c}{PointNet++}\\
     \cline{3-9}
     &&Method$\downarrow$, Metric$\xrightarrow{}$ & \textbf{ACC} (\%) & \textbf{WSR} (\%) & IoM (\%) & \textbf{ACC} (\%) & \textbf{WSR} (\%) & IoM (\%) \\
    \hline
    \multirow{7}{*}{ModelNet40} &
    Benign & No Watermark & 89.2 & N/A & N/A & 91.9 & N/A & N/A \\
    \cline{2-9}
    &
   \multirow{4}{*}{Poison-label}  
   & PCBA & 87.8 & 98.9  & \red{100} & 90.6 & 99.3 &  \red{100} \\
    && PointPBA-B & 87.6 & 99.8 & \red{100} & 90.5 & 99.5   & \red{100}\\
    && PointPBA-R &87.8 &  68.8 & \red{100} & 90.4 & 81.9 & \red{100}\\
    &&IRBA &87.8& 80.8 & \red{100} & 91.0& 88.7 & \red{100} \\
    \cline{2-9}
    &
    \multirow{2}{*}{Clean-label} 
    & PointCBA & 88.8 & \red{31.3}  & 0.0 & 91.0 & \red{50.9}  & 0.0 \\
    && PointNCBW-patch & 89.0 & \red{33.9} & 0.0 & 91.0 & \red{36.5} & 0.0 \\ 
     && PointNCBW (Ours) & 88.7 & 82.1 & 0.0 & 91.2 & 85.6& 0.0 \\ 

     \hline
    \multirow{7}{*}{ShapeNetPart} &
    Benign & No Watermark & 98.6 & N/A & N/A & 98.9& N/A & N/A \\
    \cline{2-9}
    &
   \multirow{4}{*}{Poison-label}  
   & PCBA &98.6 & 98.5  & \red{100} & 98.8 & 98.9 &  \red{100} \\
    && PointPBA-B &98.5  &100& \red{100} & 98.9 & 99.5& \red{100}\\
    && PointPBA-R & 98.3 & 84.0 & \red{100} & 98.9& 89.6  & \red{100}\\
    &&IRBA &98.4  & 92.1 & \red{100} & 98.6 & 96.9 & \red{100} \\
    \cline{2-9}
    &
    \multirow{2}{*}{Clean-label} 
    & PointCBA & 98.4 & \red{64.8}  & 0.0 & 98.4 & \red{68.6}  & 0.0 \\
    && PointNCBW-patch & 98.5 & \red{45.5}& 0.0 & 98.5  & \red{51.9}  & 0.0  \\ 
     && PointNCBW (Ours) & 98.4 & 93.1 & 0.0 & 98.5 & 97.6 & 0.0 \\ 

      \hline
    \multirow{7}{*}{PartNet} &
    Benign & No Watermark & 96.3 & N/A & N/A & 96.8 & N/A & N/A \\
    \cline{2-9}
    &
   \multirow{4}{*}{Poison-label}  
   & PCBA & 95.3&99.8  & \red{100} & 95.6 & 99.5 &  \red{100} \\
    && PointPBA-B & 95.3 & 99.9 & \red{100} & 95.5 & 99.6 & \red{100}\\
    && PointPBA-R & 95.2 &86.0  & \red{100} & 95.4 & 93.1  & \red{100}\\
    &&IRBA & 95.7 & 85.0 & \red{100} & 96.5& 90.4 & \red{100} \\
    \cline{2-9}
    &
    \multirow{2}{*}{Clean-label} 
    & PointCBA & 96.1 &\red{34.3} & 0.0 & 96.7 & \red{54.0}  & 0.0 \\
    && PointNCBW-patch & 96.2 & \red{32.3} & 0.0  & 96.7  & \red{38.7} &  0.0\\ 
     && PointNCBW (Ours) & 96.1 &  86.2 & 0.0 & 96.6 & 88.4 & 0.0 \\ 
    \bottomrule
  \end{tabular}
  }
  \label{watermark performance}
  \vspace{-0.5em}
\end{table*}

\section{Experiments}

\subsection{Experiment Setup}
\label{sec:exp_settings}

\noindent
\textbf{Datasets.}
We conduct experiments on two datasets, including ModelNet40 \cite{modelnet}, ShapeNetPart \cite{shapenet} and PartNet \cite{mo2019partnet}. 
Following \cite{pointnet}, we uniformly sample 1,024 points from the original CAD models as point clouds and normalize them into $[0, 1]^3$.

\vspace{0.3em}
\noindent
\textbf{Models.}
We adopt PointNet \cite{pointnet} as the default surrogate model for PointNCBW. We also take other common models ($e.g.$, PointNet++ \cite{pointnet++}, DGCNN \cite{dgcnn}, and PointASNL with PNL \cite{yan2020pointasnl}) into consideration. All models are trained with default settings suggested by their original papers.

\vspace{0.3em}
\noindent
\textbf{Trigger Design.} To ensure the stealthiness of our watermark, we adopt a trigger of deliberately small size. Specifically, we use a fixed sphere as our shape of trigger $\bm{\Gamma}$. We set its center as $(0.3,0.3,0.3)$ and its radius to 0.025. Our trigger consists of 50 points randomly sampled from this sphere, which takes proportion about $5\%$ of one watermarked sample (1,024 points in total). The examples of point cloud samples involved in different watermarks are shown in Figure \ref{fig:samples}.   

\vspace{0.3em}
\noindent
\textbf{Hyper-parameter.}
In shape-wise TFP, we set the number of starting points $n=30$, the number of iterations $T=30$. We use Adam optimizer \cite{kingma2014adam} to update angles $\theta$ with initial learning rate $l_r=0.025$ and $l_r$ is divided by 10 every 10 steps. In point-wise TFP, we set regularization factor $\eta=50$ in our objective function as default, and during the process of optimization, we set number of iterations $T=20$, step size $\beta=0.0025$, and decay factor $\mu=1$.

\begin{table*}[t]

    \captionsetup{font=small}
  \caption{The effectiveness of point cloud dataset ownership verification via our PointNCBW watermark under the PointNet structure.}
  \vspace{-0.3em}
  \centering
  \scalebox{1.1}{
  \begin{tabular}{c|ccc|ccc|ccc}
    \toprule
     Dataset$\rightarrow$ & \multicolumn{3}{c|}{ModelNet40} & \multicolumn{3}{c|}{ShapeNetPart} & \multicolumn{3}{c}{PartNet}\\
    \hline
    \renewcommand{\arraystretch}{0.7}
   Metric$\downarrow$, Scenario$\rightarrow$ & In-T &  In-M  & Malicious & In-T & In-M  & Malicious & In-T & In-M  & Malicious\\
    \hline
    \renewcommand{\arraystretch}{1.5}
    $\Delta P$ & 0.01 & 0.02 & 0.87 & -0.01 & 0.01 & 0.95 & 0.01 & 0.01 & 0.95\\
 p-value & 1.00 &  1.00 &$10^{-57}$ & 1.00 &  1.00 & $10^{-100}$ & 1.00 & 1.00 & $10^{-71}$\\
    \bottomrule
  \end{tabular}
  }
  \label{table: verification}
\end{table*}

\begin{table*}[t]
    \captionsetup{font=small}
  \caption{The effectiveness of point cloud dataset ownership verification via our PointNCBW watermark under the PointNet++ structure.}
\vspace{-0.3em}
  \centering
  \scalebox{1.1}{
  \begin{tabular}{c|ccc|ccc|ccc}
    \toprule
     Dataset$\rightarrow$ & \multicolumn{3}{c|}{ModelNet40} & \multicolumn{3}{c|}{ShapeNetPart}& \multicolumn{3}{c}{PartNet}\\
    \hline
    \renewcommand{\arraystretch}{0.7}
   Metric$\downarrow$, Scenario$\rightarrow$ & In-T &  In-M  & Malicious & In-T & In-M  & Malicious& In-T & In-M  & Malicious\\
    \hline
    \renewcommand{\arraystretch}{1.5}
    $\Delta P$ & 0.01 & 0.01 & 0.88 & 0.01 & 0.01 & 0.97 & 0.01 & 0.01 & 0.96\\
 p-value & 1.00 &  1.00 &$10^{-60}$ & 1.00 &  1.00 & $10^{-236}$& 1.00 &  1.00 & $10^{-108}$\\
    \bottomrule
  \end{tabular}
  }
  \label{table: verification_pn2}
  \vspace{-1.5em}
\end{table*}

\vspace{0.3em}
\noindent
\textbf{Evaluation Metrics.}
Accuracy (ACC) is used to evaluate the performance of models on benign samples. Watermark success rate (WSR) measures the effectiveness of dataset watermarks. Specifically, WSR is the percentage of verification samples that are misclassified. We use $\Delta P \in [-1,1]$ and p-value $\in[0,1]$ for ownership verification as introduced in Section \ref{Copyright Verification}.

\subsection{The Performance of Dataset Watermarking}
\label{sec:watermark_performance}

\noindent \textbf{Settings.} To evaluate the effectiveness and the stealthiness of watermarks, we compare our PointNCBW with existing poison-label and clean-label backdoor watermarks. The poison-label methods include PCBA \cite{xiang2021backdoor}, PointPBA-Ball \cite{li2021pointba}, PointPBA-Rotation \cite{li2021pointba}, and IRBA \cite{gao2023imperceptible}. We include the only existing clean-label backdoor watermark ($i.e.$, PointCBA \cite{li2021pointba}) as our main baseline. For comparative analysis, we additionally design perturbations utilizing a spherical patch to replace TFP in our PointNCBW. This approach employs a spherical patch with a radius of 0.25, where the center coordinates are optimized to minimize the relative distance between the perturbed samples and the target category (dubbed `PointNCBW-patch').  We randomly select `Keyboard' on ModelNet40, `Knife' on ShapeNetPart, and `Door' on PartNet as our target label. We set the watermarking rate as 0.01. To measure watermark stealthiness, we also calculate the percentage of samples whose label is inconsistent with the ground-truth category on modified samples (dubbed `IoM').

\vspace{0.3em}
\noindent \textbf{Results.} As shown in Table~\ref{watermark performance}, our PointNCBW watermark is significantly more effective than PointCBA, especially on the dataset containing more categories ($i.e.$, ModelNet40). For example, the watermark success rate of PointCBA is more than 50\% higher than that of PointCBA under the PointNet structure on ModelNet40. In particular, the ACC and WSR of our PointNCBW are also on par with those of poison-label watermarks that are not stealthy under human inspection (see the IoM). These results verify the benefits of our PointNCBW.



\subsection{The Performance of Ownership Verification}

\noindent \textbf{Settings.} We evaluate our proposed PointNCBW-based ownership verification in three scenarios, including \textbf{(1)} independent trigger (dubbed `In-T'), \textbf{(2)} independent model (dubbed `In-M'), and \textbf{(3)} unauthorized dataset training (dubbed `Malicious'). In the first scenario, we query the model that is trained on the watermarked dataset with the trigger that is different from the one used in the process of PointNCBW; In the second scenario, we examine the benign model which is trained on clean dataset, using the same trigger pattern in PointNCBW; Additionally, we adopt the same trigger to query the model that is trained on the watermarked dataset in the last scenario. Notice that only the last scenario should be regarded as having unauthorized dataset use. We set $\tau=0.2$ for the hypothesis test in all cases as the default setting.

\vspace{0.3em}
\noindent \textbf{Results.} As shown in Table \ref{table: verification} and Table \ref{table: verification_pn2}, our method can accurately identify unauthorized dataset usage ($i.e.$, `Malicious') with high confidence as $\Delta P$ is larger than 0.8 and p-value is nearly 0. At the same time, it avoids misjudgments when there is no stealing as the $\Delta P$ is nearly 0 and p-values are 1 under the `Independent-T' and `Independent-M' scenarios. These results show that our method can accurately identify dataset stealing without misjudgments under the black-box model access setting.



\subsection{Ablation Study}

We hereby discuss the effects of key hyper-parameters and modules of our PointNCBW. Unless otherwise specified, we exploit ModelNet40 as an example for our discussion.

\vspace{0.3em}
\noindent
\textbf{Effects of Watermarking Rate $\lambda$.}
We conduct an experiment to analyze the relationship between WSR and varying watermarking rates $\lambda$. As depicted in Figure \ref{fig:wsr}a, our WSR can reach above 80\% when watermarking rate is about 0.01. Besides, a higher watermarking rate can bring a better WSR.  

\begin{figure}[!t]
\vspace{-1em}
\centering  
\captionsetup[subfloat]{font=small}
        \subfloat[watermarking rate $\lambda$]{
		\includegraphics[scale=0.25]{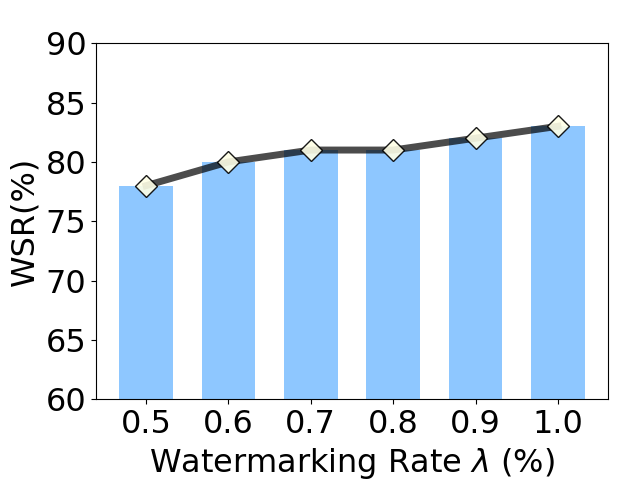}}
	  \hfill
	\subfloat[size of verification set $m$]{
		\includegraphics[scale=0.25]{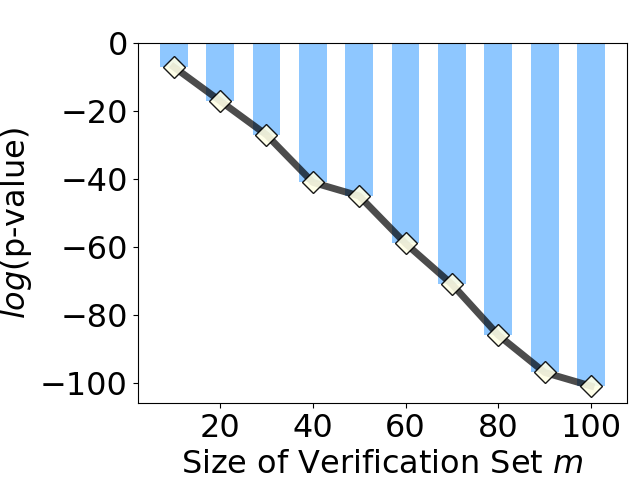}}
  \caption{Effects of watermarking rate $\lambda$ and size of verification set $m$ on the performance of PointNCBW-based dataset ownership verification on the ModelNet40 dataset.}
 \label{fig:wsr}
 \vspace{-1.5em}
\end{figure}

\begin{figure}
    \centering
    \includegraphics[width=0.75\linewidth]{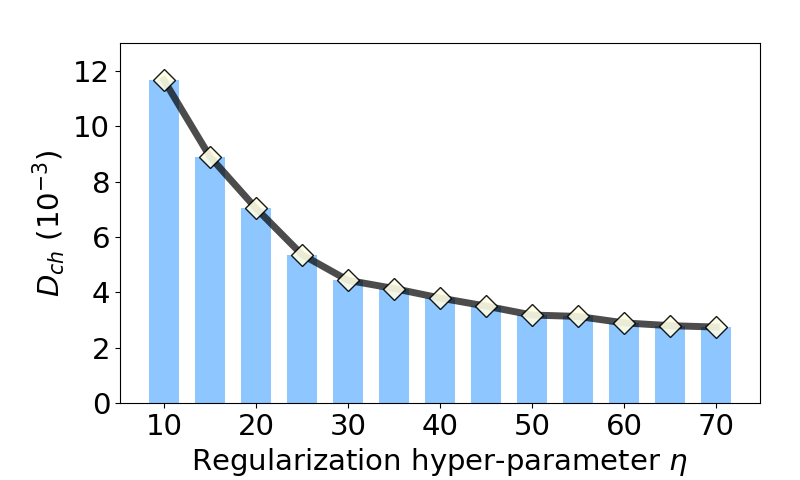}
    \caption{Effects of regularization hyper-parameter $\eta$ on the magnitude of point-wise perturbation measured by Chamfer distance ($D_{ch}$). In general, the smaller the distance, the more imperceptible the point-wise perturbations.}
    \label{fig:ch-reg}
    \vspace{-1.2em}
\end{figure}

\vspace{0.3em}
\noindent
\textbf{Effects of Size of Verification Set $m$.}
As shown in Figure \ref{fig:wsr}b, using more verification samples can significantly enhance the performance of dataset ownership verification, $i.e.$, the p-value for verification becomes smaller as the size of the verification set gets larger. The result is also consistent to our Theorem \ref{them1}.

\vspace{0.3em}
\noindent
\textbf{Effects of Regularization Hyper-parameter $\eta$.}
To evaluate the effects of regularization hyper-parameter $\eta$ in our point-wise perturbation objective function, we use Chamfer Distance \cite{achlioptas2018learning} to measure the magnitude of perturbation related to varying $\eta$. As shown in Figure \ref{fig:ch-reg}, a larger $\eta$ leads to more imperceptible point-wise perturbations. 

\vspace{0.3em}
\noindent
\textbf{Effects of Starting Points Number $n$.} In shape-wise perturbation, we adopt numerous starting points to alleviate the impact of local minima during optimization. To verify its effectiveness, we measure the relative distance (defined in Section \ref{sec: reason_explain}) between perturbed samples and target category with different numbers of starting points during shape-wise optimization. As Figure \ref{fig:startingnum} shows, a larger number of starting points during shape-wise optimization can alleviate the impact of local minima, bringing the optimized perturbed samples closer to the target category with a smaller relative distance.

\begin{figure}
    \centering
    \includegraphics[width=0.88\linewidth]{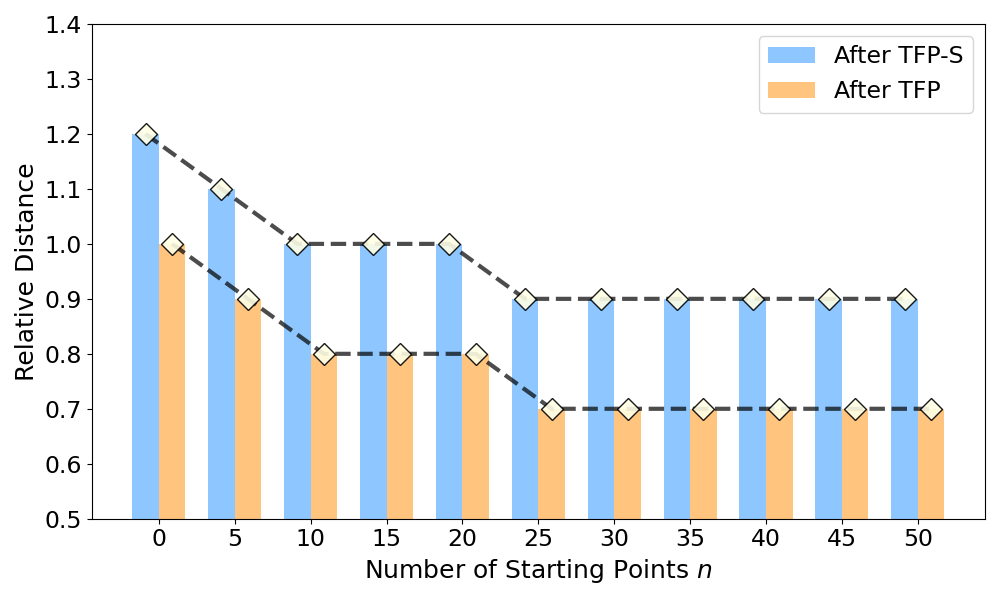}
    \caption{Effects of starting points number during shape-wise perturbation on the relative distance between final perturbed samples and target category.}
    \label{fig:startingnum}
    \vspace{-1.2em}
\end{figure}

\vspace{0.3em}
\noindent\textbf{Effects of Decay Factor $\mu$.}
We hereby analyze the impact of the decay factor $\mu$ in shape-wise optimization on the transferability of PointNCBW. In this study, we employ PointNet as the surrogate model and generate multiple versions of watermarked ModelNet40 dataset. Each of them corresponds to a different decay factor $\mu$ in the shape-wise optimization process. We train various network architectures on these watermarked datasets to assess transferability. As illustrated in Figure \ref{fig:decay}, optimal transferability is achieved when the decay factor $\mu$ is set to $1$. This setting corresponds to a cumulative gradient approach, where all previous gradients are combined to guide the optimization updates in the current iteration.

\begin{figure}
    \centering
    \includegraphics[width=0.88\linewidth]{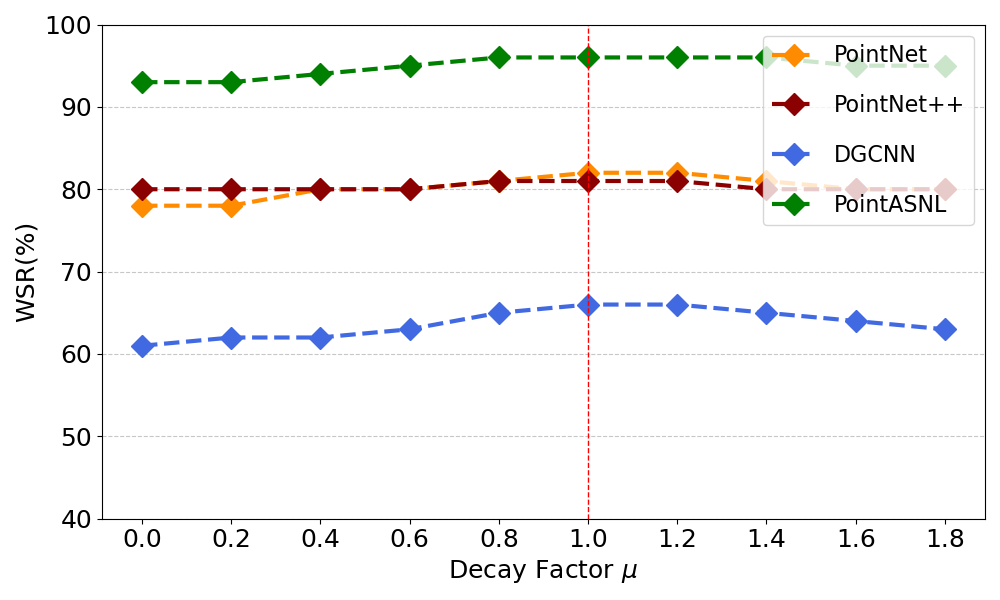}
    \caption{Effects of decay factor $\mu$ in shape-wise optimization on the transferability of PointNCBW when adopting PointNet as the surrogate network structure.}
    \label{fig:decay}
    \vspace{-1.2em}
\end{figure}

\vspace{0.3em}
\noindent
\textbf{Effects of Number of Iterations $T$.}
We assess the impact of the number of iterations $T$ in shape-wise optimization on the transferability of PointNCBW. In this study, we employ PointNet as the surrogate model and generate multiple watermarked datasets, each corresponding to a different number of iterations used in the shape-wise optimization process. Subsequently, we train various network architectures on these watermarked datasets to evaluate transferability.
As shown in Figure~\ref{fig:iter_num}, increasing the number of iterations $T$ in shape-wise optimization generally enhances transferability. However, our findings indicate that $T=20$ iterations are sufficient to achieve optimal transferability, with minimal additional benefits observed beyond this point.

\begin{figure}
    \centering
    \includegraphics[width=0.88\linewidth]{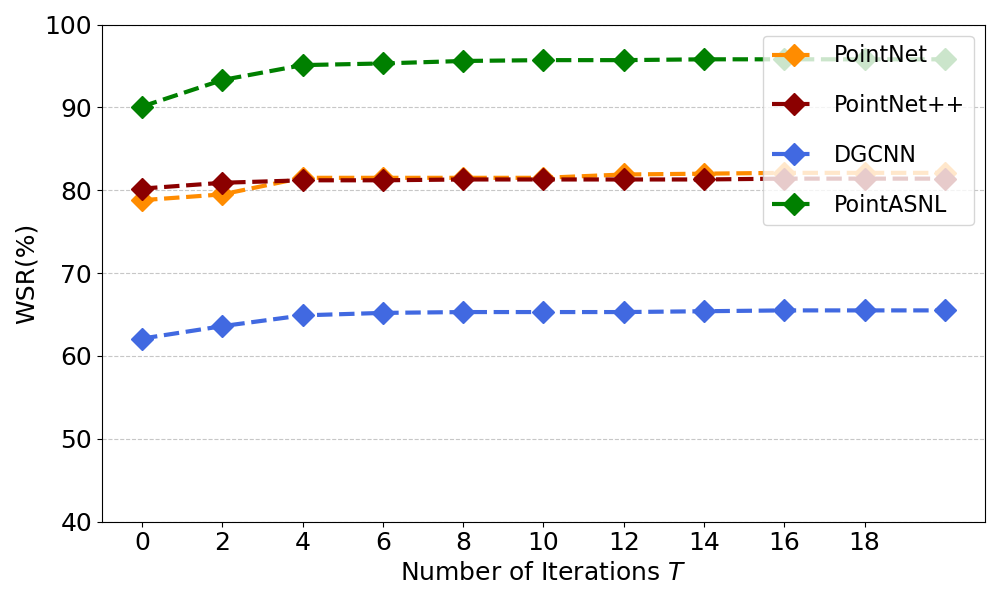}
    \caption{Effects of the number of iterations $T$ in shape-wise optimization on the transferability of PointNCBW when adopting PointNet as the surrogate network structure.}
    \label{fig:iter_num}
\end{figure}

\vspace{0.3em}
\noindent
\textbf{Effects of Verification Certainty $\tau$.}
To measure the effects of verification certainty $\tau$ in PointNCBW-based dataset ownership verification, we choose different values of the $\tau$ for ownership verification after the same watermarking process. As shown in Table \ref{tab:tau}, the p-value increases as the verification certainty $\tau$ increases across all scenarios. In particular, when $\tau$ is smaller than 0.05, our proposed PointNCBW may misjudge the cases of Independent-T or Independent-M.  In addition, the larger the $\tau$, the more unlikely the misjudgments happen and the more likely that the dataset stealing is ignored. Therefore, $\tau$ should be assigned based on specific requirements.

\begin{table}[!t]
  \captionsetup{font=small}
\caption{The p-value of PointNCBW-based dataset ownership verification $w.r.t.$ the verification certainty $\tau$ on  ModelNet40 dataset. In particular, we mark failed verification cases in \red{red}.}
  \vspace{-0.3em}
  \centering
\scalebox{1.0}{
    \begin{tabular}{c|ccccc}
    \toprule
        Scenario$\downarrow$, $\tau\rightarrow$ & 0.00 & 0.05 & 0.10 & 0.15 & 0.20 \\
        \hline
         Independent-T    & \red{$10^{-7}$} & 0.98 & 1.0 & 1.0 & 1.0  \\
         Independent-M&  \red{$10^{-11}$} & 0.47 & 1.0 & 1.0 &  1.0\\
         Malicious&   $10^{-71}$ & $10^{-68}$ & $10^{-65}$ & $10^{-60}$ & $10^{-57}$ \\
         \bottomrule
    \end{tabular}
    }
    
    \label{tab:tau}
    \vspace{-1.5em}
\end{table}

\vspace{0.3em}
\noindent\textbf{Effects of Target Category $y^{(t)}$.} 
We hereby discuss whether our PointNCBW is still effective under various target labels. We randomly choose one class as the target label while keeping all other settings unchanged. The results, as presented in Table~\ref{table targetclass}, demonstrate that although the watermark performance exhibits slight variations across different target classes, our PointNCBW consistently maintains high effectiveness.

\begin{table*}[!t]
    \captionsetup{font=small}
  \caption{The performance of our PointNCBW with different target classes under the PointNet structure.} 
 \vspace{-0.3em}
  \centering
   \scalebox{1.1}{
  \begin{tabular}{c|ccccc|ccccc}
    \toprule

Dataset$\xrightarrow{}$  & \multicolumn{5}{c|}{ModelNet40} & \multicolumn{5}{c}{ShapeNetPart} \\
    \hline
   target class$\xrightarrow{}$ & Bottle& Curtain& Door & Guitar&Wardrobe &Lamp  & Guitar & Pistol& Rocket& Skateboard\\
    
    \hline
    ACC & 88.7 & 88.8 & 88.7 & 88.7 & 88.8& 98.3 & 98.5 &98.3 &98.4&98.4\\
    WSR & 80.2 & 79.5 & 81.3 & 83.2 & 79.7&90.5& 95.2 & 90.8 &92.9& 93.6\\
    p-value & $10^{-61}$ & $10^{-56}$& $10^{-57}$& $10^{-66}$& $10^{-56}$& $10^{-78}$&$10^{-116}$&$10^{-81}$& $10^{-97}$& $10^{-105}$\\
    \bottomrule
  \end{tabular}
  }
  \label{table targetclass}
\end{table*}

\begin{table}[!t]
   \captionsetup{font=small}
\caption{The performance of PointNCBW with different trigger patterns used in the watermarking process.}
\vspace{-0.3em}
  \centering
\scalebox{1}{
    \begin{tabular}{c|ccc}
    \toprule
        Trigger Setting$\downarrow$, Metrics$\rightarrow$ & ACC & \textbf{WSR} & p-value\\
        \hline
        Original & 88.7 & 82.1 & $10^{-57}$ \\ \hline
        (a) & 88.7 & 83.0 & $10^{-57}$\\
        (b) & 88.6 & 85.1 & $10^{-60}$ \\
        (c) & 88.7 & 82.8  &  $10^{-55}$\\
        (d) & 88.7& 80.4 &$10^{-49}$ \\
    \bottomrule
    \end{tabular}
    }
    \label{tab:trigger}
\end{table}

\vspace{0.3em}
\noindent
\textbf{Effects of Trigger Patterns $\bm{\Gamma}$.} We conduct experiments on the ModelNet40 dataset to discuss the effects of trigger patterns in our PointNCBW. Specifically, we hereby discuss four main trigger settings, including \textbf{(a)} trigger pattern with different shapes, \textbf{(b)} trigger pattern with different sizes, \textbf{(c)} trigger pattern on different positions, and \textbf{(d)} trigger pattern with different number of points. In the first scenario, we sample 50 points from a cube centered at (0.3, 0.3, 0.3) with a side length of 0.05. In the second scenario, we sample 50 points from a sphere also centered at (0.3, 0.3, 0.3), but with a radius of 0.05. In the third scenario, we relocate the same default trigger to (0.3,0.3,0.3). In the last scenario, we only sample 20 points from the same sphere center at (0.3, 0.3, 0.3) with a radius of 0.025. The example of watermarked point clouds is shown in Figure \ref{fig:triggers}. As shown in Table \ref{tab:trigger}, by comparing the results of setting (a) \& (c), we know that both the shape and position of the trigger pattern used for watermarking have mild effects on the final performance. Besides, the results of setting (b) suggest that a larger trigger size leads to better watermarking and verification performance, although it may decrease the watermark's stealthiness. Furthermore, the results of setting (d) demonstrate the watermark performance may slightly decrease if the trigger contains fewer points. Nevertheless, our method obtains promising verification results across all settings.

\vspace{0.3em}
\noindent
\textbf{Effects of Feature Perturbation.}
Our TFP contains shape-wise and point-wise perturbations.  To verify effectiveness of TFP, we watermark the ModelNet40 dataset following the process of our PointNCBW under six scenarios, including \textbf{(1)} no perturbation before inserting trigger, 
\textbf{(2)} TFP with solely shape-wise part (TFP-S), \textbf{(3)} TFP with solely point-wise part (TFP-P), \textbf{(4)} TFP without momentum in shape-wise optimization (TFP-NM) and \textbf{(5)} the vanilla TFP proposed in this paper. After the processes of watermarking, we train different networks on the watermarked ModelNet40 to measure the performance of ownership verification. As shown in Table~\ref{table fp}, 
both shape-wise and point-wise perturbations are critical for the watermark and the verification performance of our PointNCBW based on TFP.

\begin{table*}[!t]
    \captionsetup{font=small}
  \caption{The performance of PointNCBW-based dataset copyright verification with different perturbations on ModelNet40.} 
  \vspace{-0.3em}
  \centering
    \scalebox{1.1}{
  \begin{tabular}{c|ccc|ccc|ccc|ccc}
 
    \toprule

      Model$\rightarrow$  & \multicolumn{3}{c|}{PointNet} & \multicolumn{3}{c|}{PointNet++} & \multicolumn{3}{c|}{DGCNN} & \multicolumn{3}{c}{PointASNL} \\
    
    \hline
    Variant$\downarrow$ & ACC & \textbf{WSR} & p-value & ACC & \textbf{WSR} & p-value & ACC & \textbf{WSR} & p-value &ACC & \textbf{WSR} & p-value \\

    \hline
    None &89.2 & 0.0 & 1.00  & 91.9 &0.0 & 1.00 & 91.7& 0.0 & 1.00 & 90.1 & 0.0& 1.00 \\
  TFP-P & 89.0 & 18.4& 0.01 & 91.7& 6.2 & 0.99 & 91.5 & 2.1& 1.00 & 89.9 & 10.6 & 0.37 \\
    TFP-S &88.8 & 78.8 & $10^{-45}$ & 91.7 & 80.2 & $10^{-49}$ & 91.1 & 62.1& $10^{-46}$ & 89.2 & 90.1 & $10^{-71}$\\
    TFP-NM &88.7 & 81.3&$10^{-55}$ & 91.5& 80.4 & $10^{-55}$ & 91.0 & 62.4 & $10^{-56}$ & 89.2 & 93.2 & $10^{-95}$\\
    Ours &88.7 & 82.1 & \bm{$10^{-57}$} & 91.5 & 81.4 & \bm{$10^{-58}$} & 91.0 &65.5 & \bm{$10^{-61}$} & 89.2 & 95.8 & \bm{$10^{-100}$} \\
    \bottomrule
  \end{tabular}
  }
  \label{table fp}
\end{table*}

\begin{table*}[t]
    \captionsetup{font=small}
  \caption{The transferability performance of our PointNCBW and PointNCBW-based dataset watermarking and ownership verification with different surrogate and training model structures on ModelNet40.} 
  \vspace{-0.3em}
  \centering
   \scalebox{1.1}{
  \begin{tabular}{c|ccc|ccc|ccc|ccc}
    \toprule

Training$\xrightarrow{}$  & \multicolumn{3}{c|}{PointNet} & \multicolumn{3}{c|}{PointNet++} & \multicolumn{3}{c|}{DGCNN} & \multicolumn{3}{c}{PointASNL} \\
    \hline
   Surrogate$\downarrow$ & ACC & \textbf{WSR} & p-value & ACC & \textbf{WSR} & p-value & ACC & \textbf{WSR} & p-value &ACC & \textbf{WSR} & p-value \\
    
    \hline
    PointNet &88.7 & 82.1 & $10^{-57}$ & 91.5 & 81.4 & $10^{-58}$ & 91.0 & 65.5 & $10^{-61}$ & 89.2 & 95.8 & $10^{-100}$ \\
    PointNet++& 88.8 & 81.4 & $10^{-59}$ & 91.3& 92.4 & $10^{-89}$ & 90.9 & 64.8 & $10^{-62}$ & 88.9& 100 & $10^{-267}$\\
    DGCNN & 89.2 & 72.4& $10^{-69}$& 91.1 & 79.3 & $10^{-47}$& 91.0 & 88.2 & $10^{-83}$ & 89.8 & 84.8 & $10^{-74}$\\
    PointASNL & 88.9 &95.2 & $10^{-120}$ & 91.2 &100 & $10^{-221}$ & 90.1 &68.9& $10^{-93}$ &89.2 &100 & $10^{-213}$\\
    \bottomrule
  \end{tabular}
  }
  \label{table transferability}
  
\end{table*}

\begin{table*}[!t]
    \captionsetup{font=small}
  \caption{The performance of PointNCBW-based dataset ownership verification under transfer-learning.} 
    \vspace{-0.3em}
  \centering
    \scalebox{1.1}{
  \begin{tabular}{c|ccc|ccc|ccc|ccc}
 
    \toprule
      Fintune$\rightarrow$  & \multicolumn{3}{c|}{ModelNet} & \multicolumn{3}{c|}{ModelNet-Part (75\%)} & \multicolumn{3}{c|}{ModelNet-Part (50\%)} & \multicolumn{3}{c}{ModelNet-Part (25\%)}\\
    \hline
    Pretrain$\downarrow$ & ACC & \textbf{WSR} & p-value & ACC & \textbf{WSR} & p-value & ACC & \textbf{WSR} & p-value & ACC & \textbf{WSR} & p-value \\
    \hline
    ShapeNetPart &88.8 & 61.3 & $10^{-36}$  &88.4 & 55.5 & $10^{-30}$ & 85.3 & 48.3 & $10^{-21}$ & 83.5 & 38.1 &  $10^{-19}$\\
    PartNet & 88.7 & 53.1 & $10^{-29}$ & 88.4 & 43.2 &$10^{-24}$ & 86.5 & 34.6 & $10^{-15}$ & 83.8 & 25.1 & $10^{-8}$  \\
  
    \bottomrule
  \end{tabular}
  }
  \label{table transferlearn}
\end{table*}

\begin{figure}[!t]
\centering  
\captionsetup[subfloat]{font=small}
        \subfloat{
		\includegraphics[scale=0.5]{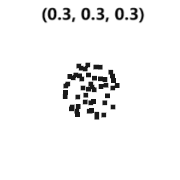}}
	  \hfill
   \addtocounter{subfigure}{-1}
	\subfloat[]{
		\includegraphics[scale=0.5]{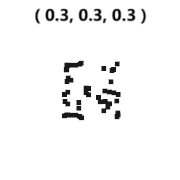}}
        \hfill
	\subfloat[]{
		\includegraphics[scale=0.5]{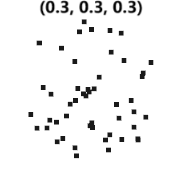}}
        \hfill
	\subfloat[]{
		\includegraphics[scale=0.5]{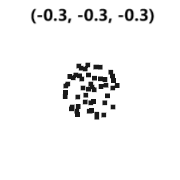}}
        \hfill
	\subfloat[]{
		\includegraphics[scale=0.5]{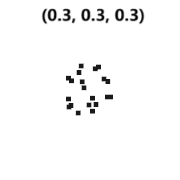}}
  \caption{The different trigger patterns used in ablation study. In particular, we have also marked the coordinates of the center point of the trigger at the top of each image.}
 \label{fig:triggers}
\end{figure}

\subsection{The Model Transferability of PointNCBW}
\label{Transferability}

Recall that our PointNCBW requires a surrogate model $g$ to generate feature perturbations, as illustrated in Eq.~(\ref{equ:perturbation general obj}). In our main experiments, we test our method under the same model structure used for generating the PointNCBW-watermarked dataset. However, the suspicious model may have a different structure compared to the one used for dataset generation in practice, since the dataset owner lacks information about the model used by dataset users. In this section, we verify that our method has model transferability and, therefore, can be used to protect dataset copyright.

\vspace{0.3em}

\noindent \textbf{Settings.} 
We exploit PointNet \cite{pointnet}, PointNet++ \cite{pointnet++}, DGCNN \cite{dgcnn}, and PointASNL \cite{yan2020pointasnl} on the ModelNet40 dataset for discussion. Specifically, we first use one of them as the surrogate model to generate the PointNCBW-watermarked dataset. After that, we also use one of them as the training model structure to train the malicious model on the generated dataset. We report the watermark and copyright verification performance of our PointNCBW on these trained models.

\vspace{0.3em}
\noindent \textbf{Results.} As shown in Table \ref{table transferability}, our method remains highly effective even when the training model is different from the surrogate one in all cases. Training networks, including PointNet \cite{pointnet}, PointNet++ \cite{pointnet++}, and PointASNL \cite{yan2020pointasnl} exhibit both high WSR and low p-value. Although training network DGCNN \cite{dgcnn} may lead to a relatively low WSR, it is still highly effective for copyright verification (\ie, p-value $\ll 0.01$). These results verify the transferability of our method. We also observe that the accuracy of TFP is relatively lower compared to other variants. This is primarily because TFP reduces the relative distance between selected samples and the target category in feature space more than the other variants, causing the trained model to become more confused about the true classification decision boundary.

\subsection{The Scalability of PointNCBW}
\label{sec Scalability}

As we demonstrated in our introduction, the performance of the only existing clean-label backdoor watermark (\ie, PointCBA \cite{li2021pointba}) is not \emph{scalable}, where its watermark performance will significantly decrease when datasets contain many categories. This limitation prevents its application as a watermark for large-scale point cloud datasets. In this section, we verify the scalability of our PointNCBW.

\vspace{0.3em}
\noindent \textbf{Settings.} We construct a series of subsets of different sizes of the original ModelNet40 dataset by randomly selecting samples from various numbers of categories ModelNet40. After that, we watermark them through our PointNCBW and train a PointNet under the same settings used in Section \ref{sec:watermark_performance}.

\vspace{0.3em}
\noindent \textbf{Results.} We compare the scalability between PointCBA and our PointNCBW in Figure \ref{fig:scalable}. The results indicate that the WSR of PointCBA is significantly degraded when the dataset contains more categories. For example, the WSR drops even below 40\% when the number of categories exceeds 30. In contrast, our PointNCBW can maintain a high WSR ($>80\%$) with the increase of the number of categories. These results verify the effectiveness of our negative trigger design proposed in PointNCBW for its scalability.


\begin{figure}[!t]
\vspace{-1em}
\centering  
\captionsetup[subfloat]{font=small}
        \subfloat[PointCBA]{
		\includegraphics[scale=0.2]{figures/wsr_pointcba.png}}
	  \hfill
	\subfloat[PointNCBW (Ours)]{
	\includegraphics[scale=0.2]{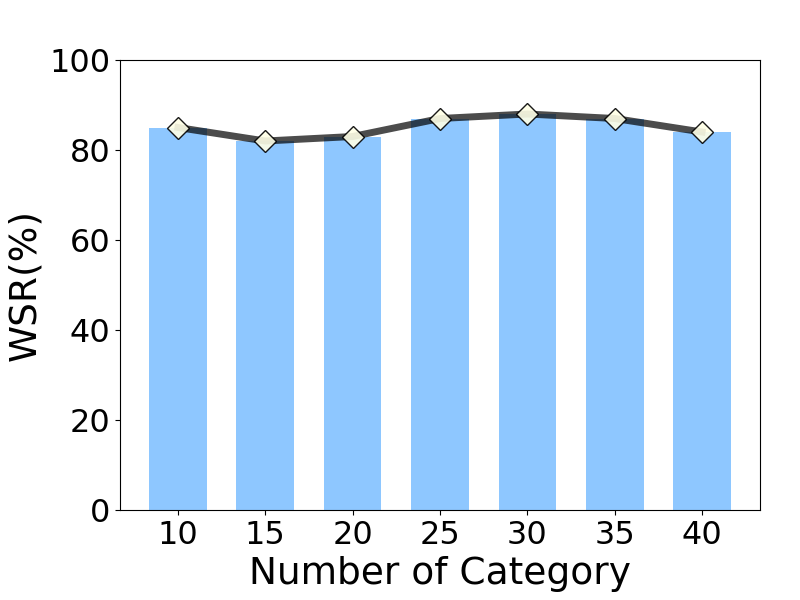}}
  \caption{The scalability comparison between PointCBA and PointNCBW. It is measured by the watermark success rate $w.r.t.$ the number of categories contained in the victim dataset.}
 \label{fig:scalable}
 \vspace{-1.5em}
\end{figure}

\subsection{PointNCBW under Transfer Learning}
In previous parts, we focused only on scenarios where adversaries directly exploit the protected dataset to train a model from scratch. As transfer learning \cite{zhuang2020comprehensive} becomes increasingly common in model training, we discuss whether our method is still effective in the transfer learning scenario.

\vspace{0.3em}
\noindent \textbf{Settings.} We hereby conduct experiments on ShapeNetPart \cite{shapenet}, PartNet \cite{mo2019partnet}, and ModelNet40 \cite{modelnet} datasets with PointNet \cite{pointnet} as the model architecture. We implement a two-stage training process. Specifically, we first pre-train PointNet on ShapeNetPart and PartNet datasets. After that, we fine-tune these pre-trained models using our watermarked ModelNet40 dataset. This fine-tuning is conducted on both the complete ModelNet40 dataset and its various subsets with different sample proportions. Unless otherwise specified, all settings are the same as those stated in Section \ref{sec:exp_settings}. 

\vspace{0.3em}
\noindent \textbf{Results.} As shown in Table \ref{table transferlearn}, although the watermark success rate (WSR) of our PointNCBW may decrease under transfer learning conditions, it remains highly effective for dataset copyright verification. We note that the partial ModelNet40 dataset may not be adequate for the model to fully capture the distinguishing features of different categories, leading to comparatively lower accuracy for PointNet when trained on ModelNet subsets. For instance, the accuracy on ModelNet is 88.8\%, while on ModelNet-Part (25\%) it is 83.5\%. Besides, we observe that our method has better performance on models pre-trained on the PartNet dataset demonstrate greater compared to those pre-trained on ShapeNetPart, as the model pre-trained on PartNet achieves a relatively higher WSR than the one pre-trained on ShapeNetPart. This phenomenon may be attributed to the larger scale of PartNet, which enables the acquisition of more comprehensive prior knowledge during pre-training.

\subsection{Stealthiness and Watermark Performance}

Compared to poison-label backdoor methods, our proposed PointNCBW is more hidden and harder to detect by humans, as it maintains the consistency between point clouds and their labels. However, we acknowledge that the watermark created by PointNCBW may still be detectable under extremely detailed human inspection. However, we believe that in practice, detailed detection is nearly impossible on large-scale datasets because of the high cost and time required.
In this section, we analyze the trade-off between the stealthiness and performance of the watermark in our PointNCBW method.

\vspace{0.3em}
\noindent \textbf{Settings.} To control the magnitude of perturbation, we implement a systematic approach utilizing two key parameters: varying the starting points $n$ in shape-wise perturbation optimization and the regularization hyper-parameter $\eta$ in point-wise perturbation for the generation of watermarked samples. In our experimental setup, we employ PointNet as the surrogate model. We then generate a series of watermarked versions of the ModelNet40 dataset based on our PointNCBW, each corresponding to different combinations of $n$ and $\eta$ values. Subsequently, we train separate instances of PointNet on these watermarked datasets to evaluate the watermark performance under varying perturbation conditions.

\vspace{0.3em}
\noindent \textbf{Results.} As evidenced by the results presented in Table~\ref{tab: wsr_stealth}, an increased magnitude of perturbation, characterized by larger values of $n$ and lower values of $\eta$, correlates with enhanced watermark performance. Visual representations of the watermarked samples are provided in Figure \ref{Fig wsr_stealth} for qualitative assessment. It is noteworthy that larger values of $n$ may still remain inconspicuous, as the resultant rotations are inherently subtle. Conversely, lower values of $\eta$ result in more noticeable perturbations, as evidenced by the increasing Chamfer Distance ($D_{ch}$) \cite{achlioptas2018learning} between the samples after shape-wise perturbation and their subsequent point-wise perturbed counterparts. The choice of perturbation magnitude depends on the specific needs and constraints of the dataset owner.

\begin{figure*}[!t]
    \centering
    \includegraphics[width=0.9\linewidth]{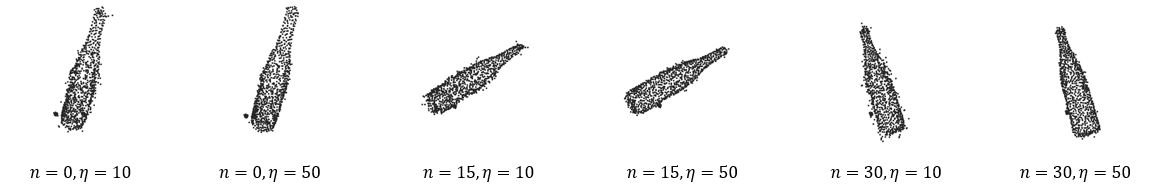}
    \caption{The watermark sample of 'Bottle' generated by varying magnitude of perturbation.}
    \label{Fig wsr_stealth}
\end{figure*}

\begin{table}[!t]
  \captionsetup{font=small}
\caption{The performance of PointNCBW with varying magnitudes of the perturbation.}
\vspace{-0.3em}
  \centering
\scalebox{1.0}{
    \begin{tabular}{c|cccccc}
    \toprule
        n$\rightarrow$ & 0 & 0 & 15 & 15 & 30 & 30\\
        \hline
        $\eta$$\rightarrow$ & 10 &50  & 10 & 50 & 10 &50 \\
        \hline
          WSR & 26.2 & 17.9 & 67.5 & 58.0 & 86.2 & 82.1\\
         p-value & $10^{-28}$ & $10^{-19}$ & $10^{-45}$ & $10^{-38}$ & $10^{-63}$ & $10^{-57}$\\
         $D_{ch}$ & 0.012 & 0.004 & 0.011 & 0.004 & 0.012 & 0.004\\
         \bottomrule
    \end{tabular}
    }
    \label{tab: wsr_stealth}
\end{table}

\subsection{The Resistance to Potential Watermark-removal Attacks}

Once malicious dataset users learn that the owners may watermark their dataset via our PointNCBW, they may design watermark-removal attacks to bypass our watermark. This section exploits ModelNet40 as an example to evaluate whether our method resists them. We consider the most widely used and representative watermark-removal methods, including data augmentation \cite{krizhevsky2012imagenet}, outlier detection \cite{rousseeuw2005robust}, and model fine-tuning \cite{devlin2018bert}, for discussion. We also design an adaptive method to further evaluate it under the setting that adversaries know our watermarking method but do not know specific settings.

\vspace{0.3em}
\noindent
\textbf{Data Augmentation.} Data augmentation is a widely used technique to enhance the diversity and richness of training data by applying random transformations or expansions to the original data. It aims to improve the generalization ability and performance of models. 
Our data augmentation methods consist of \textbf{(1)} randomly rotating the point cloud sample alongside the Eluer angles ranging ($-180\degree$, $180\degree$) and \textbf{(2)} adding Gaussian noise with mean $\mu=0$, variance $\sigma=0.01$ to point cloud sample. The results in Table \ref{table resist} demonstrate that our PointNCBW can resist common augmentation methods. 

\vspace{0.3em}
\noindent
\textbf{Outlier Removal.} Statistical outlier removal aims to identify and eliminate data points that significantly deviate from the expected or typical pattern in a dataset. Outliers are notably distant from the majority of the data, which may be used to detect trigger patterns. We perform statistical outlier removal (SOR) \cite{rousseeuw2005robust} on the generated watermarked dataset. Specifically, we compute the average distance for a given point using its 20 nearest neighbors and set a threshold based on the standard deviation of these average distances. It can be observed from Figure \ref{fig:sor} that no matter the threshold ranging from 0.5 to 2.0, the SOR fails to detect and remove our trigger sphere. This is mostly because the density of the trigger pattern is approximate with or less than the remaining parts.

\begin{figure}[!t]
\centering  
 \captionsetup[subfloat]{font=small}
        \subfloat[threshold=0.5]{
		\includegraphics[scale=0.25]{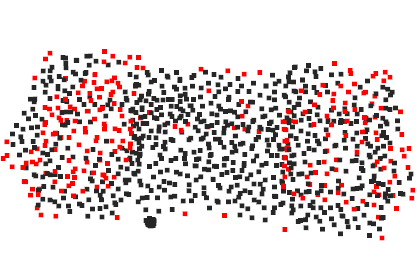}}
  \hspace{2.8em}
	\subfloat[threshold=1.0]{
		\includegraphics[scale=0.25]{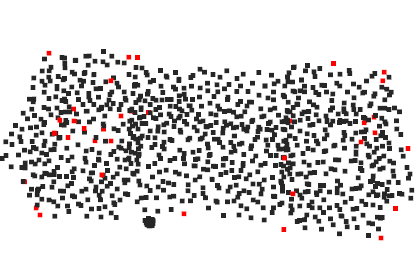}}
  
  \subfloat[threshold=1.5]{
		\includegraphics[scale=0.25]{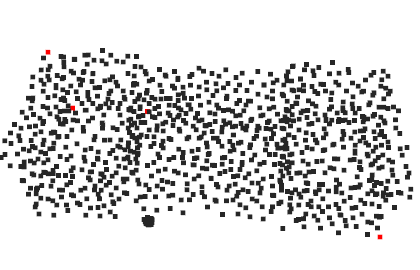}}
     \hspace{2.8em}
        \subfloat[threshold=2.0]{
		\includegraphics[scale=0.25]{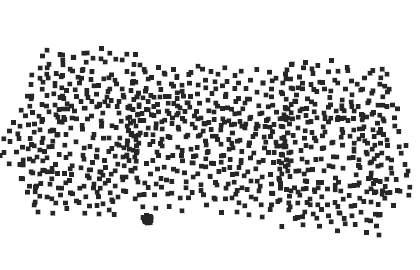}}
	\caption{The results of statistical outlier removal on a watermarked sample. The detected outlier points are marked in red.}
 \label{fig:sor}
 \vspace{-0.8em}
\end{figure}

\vspace{0.3em}
\noindent
\textbf{Fine-tuning.} We hereby evaluate the PointNCBW-watermarked model after fine-tuning, which is also a common strategy for removing potential watermarks. Initially, we train PointNet on the watermarked ModelNet40 dataset for 200 epochs. Subsequently, we randomly select 20\% of the benign samples to continue training the model for an additional 200 epochs. As shown in Table \ref{table resist}, our method is still highly effective (p-value $\ll 0.01$), although FT can slightly increase the p-value.

\vspace{0.3em}
\noindent
\textbf{Adaptive Adversarial Removal.} 
We assume that malicious dataset users have prior knowledge of the existence of the PointNCBW watermark in the dataset. Specifically, they understand how our PointNCBW works but lack access to the exact target labels, the trigger pattern, and the specific watermarked samples employed by the dataset owners in PointNCBW. In this scenario, we design an adaptive removal method that might be used by malicious dataset users. The truth is that our PointNCBW relies much on feature perturbation, and we have experimentally proved that the closer feature distance between selected non-target samples and target category can lead to better watermark performance in Section \ref{sec: reason_explain}, and the reverse is also true. Consequently, we design to disentangle the features of each sample adversarially during the training process. Specifically, we train PointNet \cite{pointnet} on watermarked ModelNet40 \cite{modelnet} with 200 epochs, and on every 10 epochs, we rotate each sample in the training dataset to bring its feature away from the current feature as far as possible. Let $\bm{x}_{tr}$ denotes one sample during training phase, we rotate $\bm{x}_{tr}$ with rotation matrix $S(\theta')$, where 
\begin{equation}
\begin{aligned}
\label{adaptive}
\theta' = \mathop{\arg\max}\limits_{\theta'}\ {\mathcal{E}(g_f(\bm{x}_{tr}\cdot S(\theta')),g_f(\bm{x}_{tr})).}
\end{aligned}
\end{equation}
We approximately optimize Eq.~(\ref{adaptive}) in a method similar to Algorithm \ref{alg:sfe}, except that the $\theta'$ is updated in the same direction as the gradient. The results in Table \ref{table resist} show that the adaptive method is significantly more effective compared to other watermark-removal attacks. However, our method is still highly effective with a high WSR and low p-value. In other words, our PointNCBW-based dataset ownership verification is also resistant to this adaptive attack.

\begin{table*}[t!]
    \captionsetup{font=small}
  \caption{The performance of our PointNCBW under different potential watermark-removal attacks on the ModelNet40 dataset.}
\vspace{-0.3em}
\centering
\scalebox{1.0}{
  \begin{tabular}{c|c|c|c|c|c|c}
  
    \toprule
   Metric$\downarrow$, Method$\rightarrow$ & No Attack & Augmentation (Rotation) & Augmentation (Noise) & Outlier Removal & Fine-tuning & Adaptive Removal\\
   \hline
  ACC & 88.7& 88.2& 88.4 & 88.7 & 88.8 & 88.5\\
  WSR & 82.1& 78.2 & 79.4 & 82.0 & 81.5 & 75.6 \\
  p-value & $10^{-57}$ & $10^{-41}$& $10^{-45}$ & $10^{-57}$ & $10^{-55}$ &$10^{-39}$\\
  \bottomrule
   
  \end{tabular}
  }
  \label{table resist}
\end{table*}

\begin{figure}[!t]
\vspace{-1em}
	\centering  
 \captionsetup[subfloat]{font=small}
	\subfloat[before watermarking]{
		\includegraphics[scale=0.25]{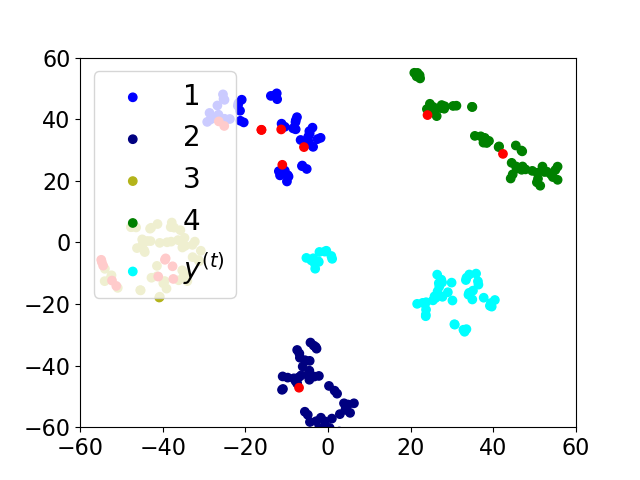}}
	\subfloat[after watermarking]{
		\includegraphics[scale=0.25]{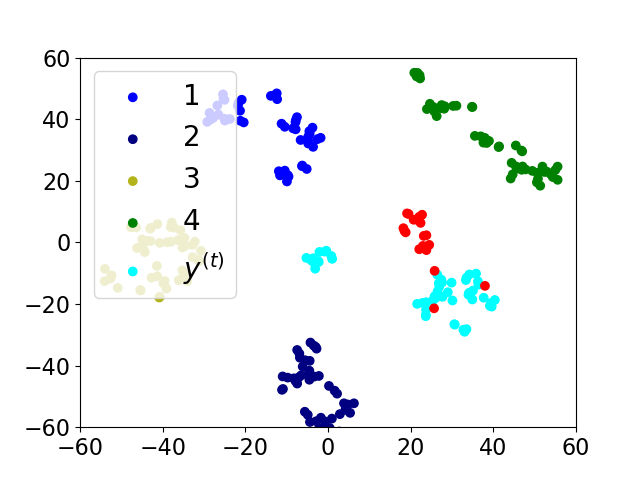}}
\caption{The t-SNE visualization. \textbf{(a)} and \textbf{(b)} depict the feature space of samples during our PointNCBW. The selected samples for modification are marked in red.}
 \label{fig: feature}
 \vspace{-1em}
\end{figure}

\begin{figure}
    \centering
    \includegraphics[width=0.7\linewidth]{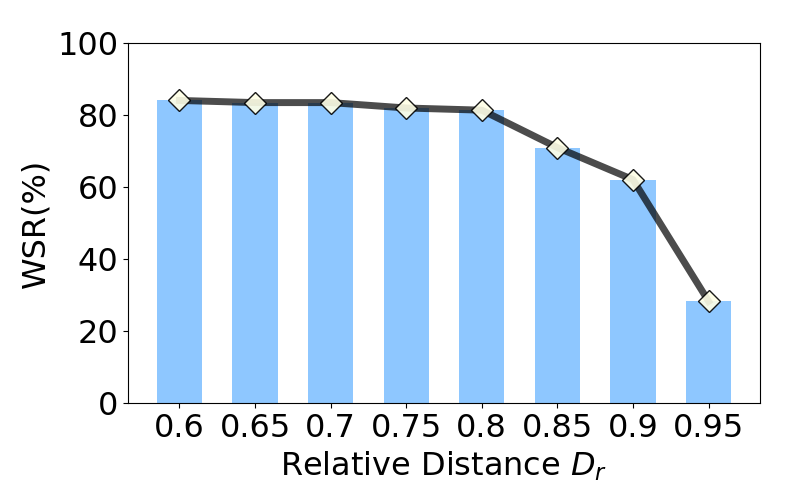}
    \caption{Effects of the relative distance between the modified samples and the target label on WSR.}
    \label{fig: wsr_dr}
    \vspace{-0.8em}
\end{figure}

\subsection{Why Is Our PointNCBW Highly Effective?}

\label{sec: reason_explain}

To investigate why our method is highly effective, we first visualize the features of samples before and after watermarking. Specifically, we randomly select some samples from five different categories, including $y^{(t)}$, and project their features into 2D space by the t-SNE method \cite{van2008visualizing}. As shown in Figure \ref{fig: feature}, features of selected samples (marked in red) were distributed over all categories before watermarking. In contrast, they move closer to category $y^{(t)}$ after our PointNCBW. Based on the feature shift, the trained models can discover that all these samples share the same part ($i.e.$, trigger $\bm{\Gamma}$). Accordingly, the models will explain the reason why these samples have similar features as category $y^{(t)}$, but different true labels might be attributed to the existence of trigger $\bm{\Gamma}$. Consequently, the trained model will interpret our trigger as one key component to deny predicting label $y^{(t)}$. 

For further study, we also calculate the relative distance (\ie, $D_r$) between watermarked sample $\bm{x}_m$ and target samples in feature space, as follows:
\begin{equation}
\label{relative distance}
D_r=\frac{\Vert \Bar{g_f}(x_m)-\Bar{g_f}(x_t) \Vert_2}{\Vert \Bar{g_f}(x_t)\Vert_2},
\end{equation}
where $\Bar{g_f}$ is mean of feature representations. As shown in Figure \ref{fig: wsr_dr}, the WSR increases with the decrease of $D_r$. It verifies the effectiveness of our TFP.

Besides, we hereby also present more examples of watermarked samples of our PointNCBW to further verify our watermark stealthiness, as shown in Figure \ref{fig:samples_more}. Note that dataset owners can also exploit other trigger patterns to further increase the stealthiness based on the characteristics of their victim dataset. It is out of the scope of this paper.

\begin{figure*}[!t]
    \centering
    \includegraphics[width=0.8\linewidth]{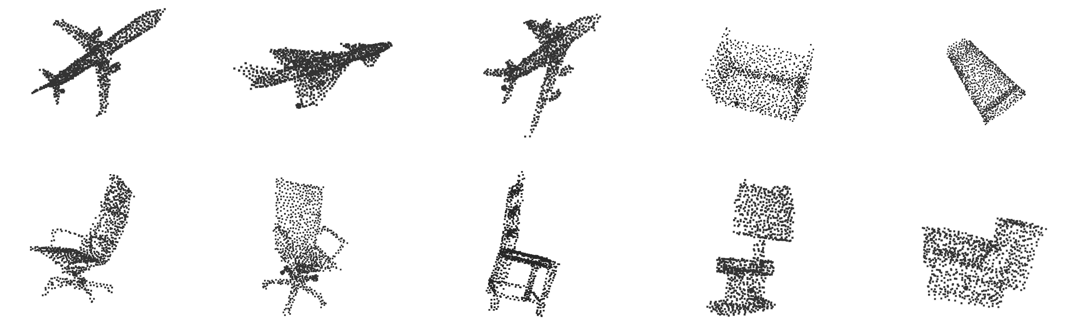}
    \vspace{0.2em}
    \caption{The example of watermarked samples generated by our PointNCBW.}
    \vspace{-1em}
    \label{fig:samples_more}
\end{figure*}

\subsection{The Analysis of Computational Complexity}
In this section, we analyze the computational complexity of our PointNCBW. We hereby discuss the computational complexity of dataset watermarking and verification, respectively.

\vspace{0.3em}
\noindent \textbf{The Complexity of Dataset Watermarking.} Let $N$ denote the size of the original dataset, our computational complexity is $\mathcal{O}(\lambda\cdot N)$ since PointNCBW only needs to watermark a few selected samples ($N$ 
 times watermarking rate $\lambda$). Besides, PointNCBW supports parallel processing, as the TFP on the selected samples is independent. For example, to apply PointNCBW-based watermark on each selected sample when using PointNet as the surrogate model, it takes about one second and 1GB memory on NVIDIA GeForce RTX 3090. As such, the additional time required by our method in the watermarking phase is negligible to a large extent. 

\vspace{0.3em}
\noindent \textbf{The Complexity of Dataset Verification.} After obtaining the API of one suspicious black-box model, we use $m$ verification samples to test whether the model is trained on our watermarked dataset. The computational complexity of this step is $\mathcal{O}(m)$, and it also supports parallel processing. For example, it takes about 50ms to use 100 verification samples (in batch mode) for ownership verification. Accordingly, the additional time our method requires in the verification phase is also negligible to a large extent.

\section{Potential Limitations and Future Directions}

In our PointNCBW, we design transferable feature perturbation (TFP) to bring selected non-target samples closer to the target category in the feature space. We have empirically proven that our PointNCBW-based verification succeeds only when the feature distances between selected perturbed non-target samples and the target category are close enough. However, TFP heavily relies on rotation transformation and is effective only on models that exhibit sensitivity to rotation. Current point cloud models are sensitive to geometric transformations, primarily due to the use of max-pooling operations \cite{pointnet, pointnet++, dgcnn, yan2020pointasnl}, which contributes to the efficacy of our PointNCBW. Nevertheless, our approach could be less effective if future point cloud model architectures achieve greater robustness to rotation transformations. 
To address this, we aim to develop adaptive perturbation mechanisms that do not rely on rotational sensitivity. We will investigate more powerful mechanisms in our future work.

Secondly, our work is primarily centered on digital environments, whereas point cloud models are increasingly likely to be deployed in complex real-world scenarios, such as autonomous driving. Verification in such real-world contexts may pose greater challenges, as the placement and scale of the trigger may not be optimally configured. We will further explore it in our future work.

Thirdly, our PointNCBW will implant distinctive yet stealthy backdoor behaviors into the trained model. Similar to all existing dataset ownership verification (DOV) methods, the embedded backdoors may be maliciously used by adversaries, although our PointNCBW intends to protect dataset copyright instead of for attack. However, since we adopt a negative trigger pattern that aims only to misclassify samples from the target class, PointNCBW has minor threats due to its untargeted nature and limited misclassification behaviors. We will explore how to design purely harmless DOV methods for point clouds in our future work.

\section{Conclusion}
In this paper, we conducted the first attempt to protect the copyright of point cloud datasets by backdoor-based dataset ownership verification (DOV). We revealed that existing backdoor watermarks were either conspicuous or not scalable to large datasets due to their positive trigger effects. To alleviate this problem, we proposed a simple yet effective clean-label backdoor watermark for point clouds by introducing negative trigger effects. Specifically, we performed transferable feature perturbation (TFP) on non-target samples before implanting the trigger, aiming to bring perturbed non-target samples closer to the target category in feature space. Consequently, the trained model will view our trigger as a signal to deny predicting the target category. We also designed a hypothesis-test-guided dataset ownership verification via our watermark and provided its theoretical analyses. We hope our work provides a new angle for creating backdoor watermarks to facilitate trustworthy dataset sharing and trading.

\section*{Acknowledgment}
This research is supported in part by the National Key Research and Development Program of China under Grant 2021YFB3100300, the National Natural Science Foundation of China under Grants (62072395 and U20A20178), and the Key Research and Development Program of Zhejiang under Grant 2024C01164. This work was mostly done when Yiming Li was a Research Professor at the State Key Laboratory of Blockchain and Data Security, Zhejiang University. He is currently at Nanyang Technological University.

\bibliographystyle{plain}
\bibliography{main}

\appendix


\setcounter{theorem}{0}

\begin{theorem}\label{them1_ap}
Suppose $f$ generates the posterior probability by a suspicious model. Let variable $X$ denote the benign sample from the target class $y^{(t)}$ and variable $X'$ is its verified version. Let $P_b = f(X)_{y^{(t)}}$ and $P_v = f(X')_{y^{(t)}}$ denote the predicted probability of $X$ and $X'$ on $y^{(t)}$. Assume that $P_b > \zeta$, we claim that dataset owners can reject the null hypothesis $H_0$ at the significance level $\alpha$, if the watermark success rate $W$ of $f$ (with $m$ verification samples) satisfies:

\begin{equation}
    \sqrt{m-1} \cdot (W + \zeta - \tau - 1) - t_{\alpha} \cdot \sqrt{W-W^2} > 0,
\end{equation}
where $t_{\alpha}$ is $\alpha$-quantile of t-distribution with $(m-1)$ degrees of freedom.
\end{theorem}

\begin{proof}
Since $\boldsymbol{P}_b>\zeta$, the original hypothesis $H_1$ can be converted to

\begin{equation}
H_1^{\prime}:\left(1-P_v\right)+(\zeta-\tau-1)>0 .
\end{equation}

Suppose $C$ is the classifier of $f$, \textit{i.e.}, $C=\arg \max f$. Let $E$ denote the event of whether the suspect model $f$ predicts a watermark sample as the target label $y^{(t)}$. As such, $E \sim B(1, p)$, where $1-p=1-\operatorname{Pr}\left(C\left(\boldsymbol{X}^{\prime}\right)=y^{(t)}\right)=\operatorname{Pr}\left(C\left(\boldsymbol{X}^{\prime}\right) \neq y^{(t)}\right)$ indicates the verification success probability and $B$ is the Binomial distribution.

Let $\hat{\boldsymbol{x}}_1, \cdots, \hat{\boldsymbol{x}}_m$ denotes $m$ watermarked samples used for dataset verification via our PointNCBW and $E_1, \cdots, E_m$ denote their prediction events, we know that the $W$ satisfies
\begin{equation}
\begin{aligned}
& W=\frac{1}{m} \sum_{i=1}^m\left(1-E_i\right), \\
& W \sim \frac{1}{m} B(m, 1-p) .
\end{aligned}
\end{equation}

According to the central limit theorem, the watermark success rate $W$ follows Gaussian distribution $\mathcal{N}\left(1-p, \frac{p(1-p)}{m}\right)$ when $m$ is sufficiently large. Similarly, $\left[\left(1-P_v\right)+(\zeta-\tau-1)\right]$ also satisfies Gaussian distribution. Accordingly, we can construct the $t$-statistic as follows:
\begin{equation}
\begin{aligned}
T \triangleq \frac{\sqrt{m} \cdot (W+\zeta-\tau-1)}{s} \sim t(m-1),
\label{eq:1thrm}
\end{aligned}
\end{equation}
where $s$ is the standard deviation of $W+\zeta-\tau-1$ and $W$, \ie,
\begin{equation}
\begin{aligned}
s^2=\frac{1}{m-1} \sum_{i=1}^m\left(E_i-W\right)^2=\frac{1}{m-1}\left(m \cdot W-m \cdot W^2\right) .
\label{eq:2thrm}
\end{aligned}
\end{equation}

To reject $H_0$ at the significance level $\alpha$, we have:
\begin{equation}
\frac{\sqrt{m} \cdot(W+\zeta-\tau-1)}{s}>t_\alpha,
\end{equation}
where $t_\alpha$ is the $\alpha$-quantile of $t$-distribution with $(m-1)$ degrees of freedom. According to Eq.~(\ref{eq:1thrm})\&Eq.~(\ref{eq:2thrm}), we have
\begin{equation}
\sqrt{m-1} \cdot(W+\zeta-\tau-1)-t_\alpha \cdot \sqrt{W-W^2}>0.
\end{equation}
\end{proof}

\begin{IEEEbiography}[{\includegraphics[width=1in,height=1.25in,clip,keepaspectratio]{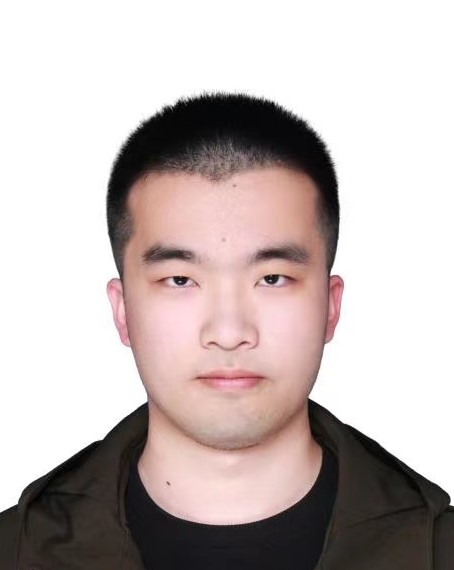}}]{Cheng Wei} received the bachelor’s degree from the
School of Computer Science and Technology, Harbin Engineering University in 2022. He is currently pursuing a Master's degree with the School of Cyber Science and Technology, at Zhejiang University. His research interests include trustworthy AI and computer vision.
\end{IEEEbiography}

\begin{IEEEbiography}[{\includegraphics[width=1in,height=1.25in,clip,keepaspectratio]{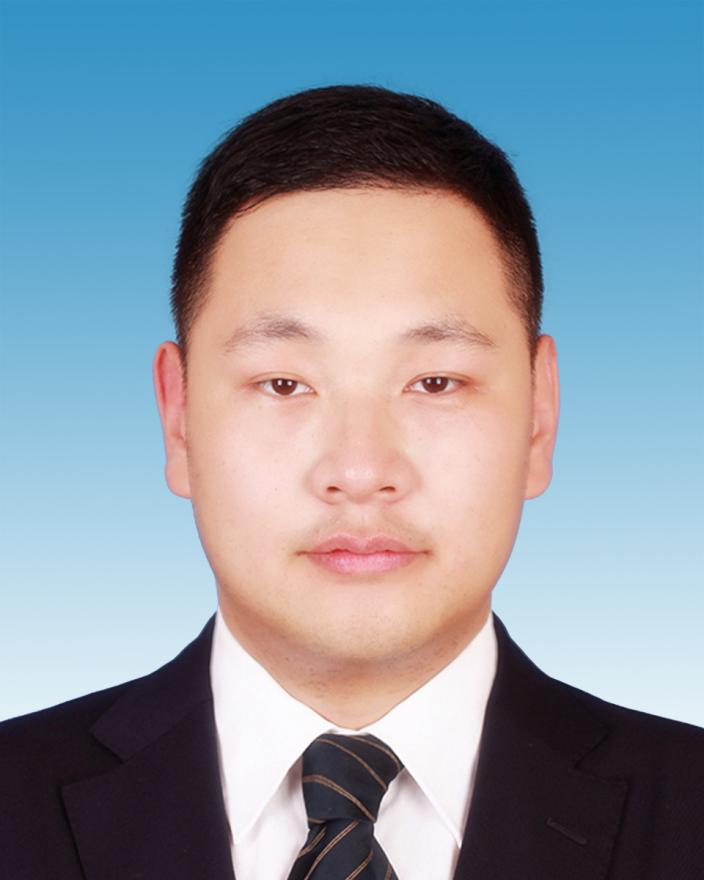}}]{Yang Wang} received his B.S. degree in communication engineering from Zhejiang University and Master degree in System on chip from KTH. He is currently a senior engineer in Hangzhou Institute for Advanced Study, UCAS, and working toward the PhD degree in Zhejiang University. His research interests include data security and privacy.
\end{IEEEbiography}

\begin{IEEEbiography}[{\includegraphics[width=1in,height=1.25in,clip,keepaspectratio]{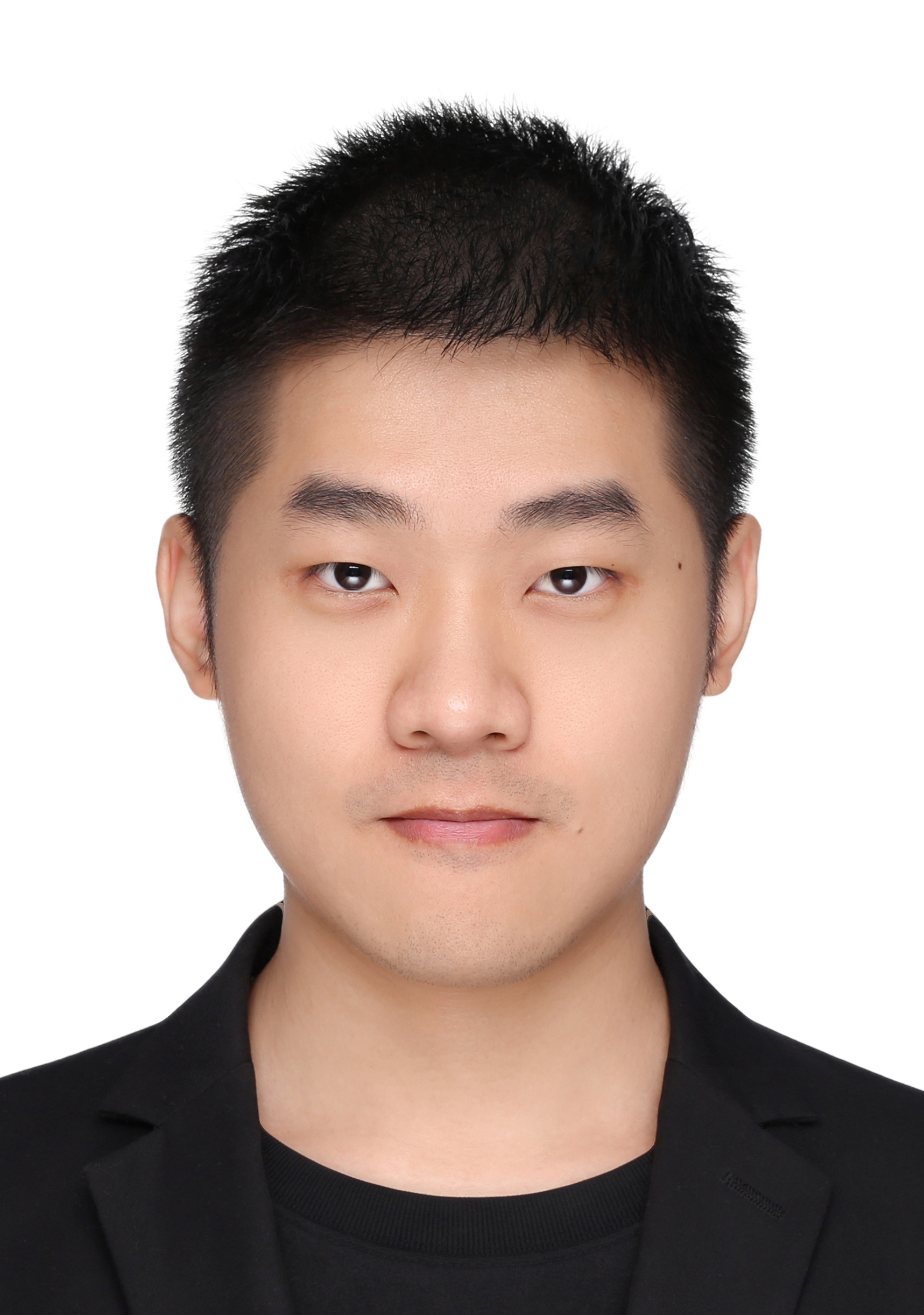}}]{Kuofeng Gao} received the bachelor’s degree from the School of Cyber Science and Engineering, Wuhan University, Wuhan, China, in 2021. He is currently pursuing the Ph.D. degree with the Shenzhen International Graduate School, Tsinghua University, Guangdong, China. His research interest generally includes trustworthy and responsible AI.\end{IEEEbiography}

\begin{IEEEbiography}[{\includegraphics[width=1in,height=1.25in,clip,keepaspectratio]{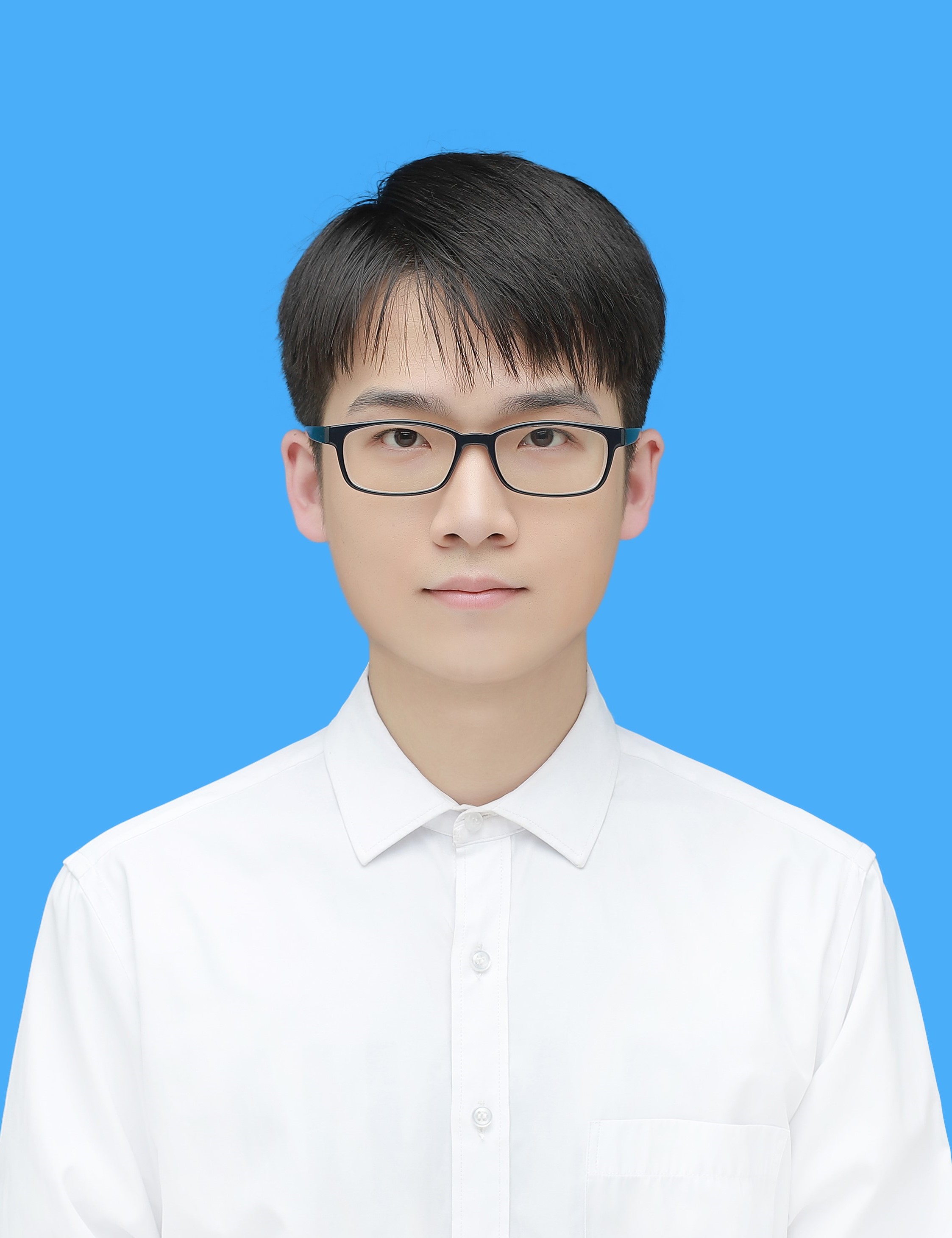}}]{Shuo Shao} received the B.Eng. degree from the
School of Computer Science and Technology, Central South University in 2022. He is currently pursuing a Ph.D. degree with the School of Cyber
Science and Technology and the State Key Laboratory of Blockchain and Data Security, Zhejiang University. His
research interests include copyright protection in AI, backdoor attack and defense, and AI security.
\end{IEEEbiography}

\begin{IEEEbiography}[{\includegraphics[width=1in,height=1.25in,clip,keepaspectratio]{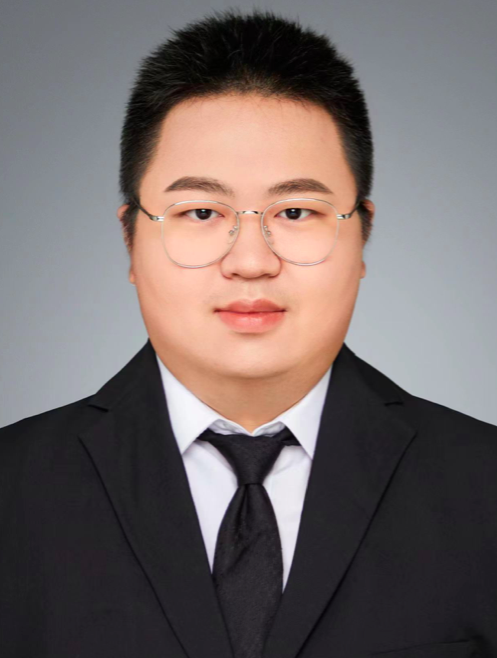}}]{Yiming Li} (Member, IEEE) received the B.S. degree
(Hons.) in mathematics from Ningbo University in 2018 and the Ph.D. degree (Hons.) in computer science and technology from Tsinghua University in 2023. He is currently a Research Fellow with the College of Computing and Data Science, Nanyang Technological University. Before that, he was a Research Professor with the State Key Laboratory of Blockchain and Data Security, Zhejiang University, and HIC-ZJU. His research has been published in multiple top-tier conferences and journals, such as IEEE S\&P, NDSS, ICLR, NeurIPS, ICML, and IEEE TRANSACTIONS ON INFORMATION FORENSICS AND SECURITY. His research interests include the domain of trustworthy ML and responsible AI. His research has been featured by major media outlets, such as IEEE Spectrum. He was a recipient of the Best Paper Award at PAKDD in 2023 and the Rising Star Award at WAIC in 2023. He served as the Area Chair for ACM MM and the Senior Program Committee Member for AAAI.
\end{IEEEbiography}

\begin{IEEEbiography}[{\includegraphics[width=1in,clip,keepaspectratio]{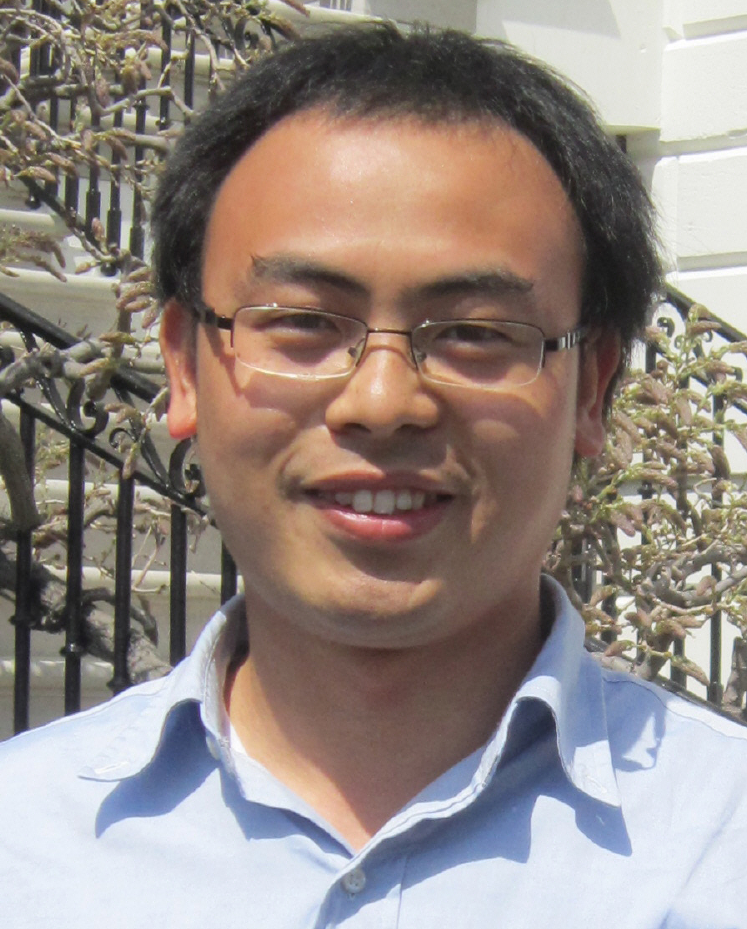}}]{Zhibo Wang}
received the B.E. degree in Automation from Zhejiang University, China, in 2007, and his Ph.D degree in Electrical Engineering and Computer Science from University of Tennessee, Knoxville, in 2014. He is currently a Professor with the School of Cyber Science and Technology, Zhejiang University, China. His currently research interests include AI security, Internet of Things, data security and privacy. He is a Senior Member of IEEE and ACM.
\end{IEEEbiography}

\begin{IEEEbiography}[{\includegraphics[width=1in,height=1.25in,clip,keepaspectratio]{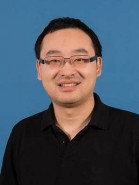}}]{Zhan Qin} is currently a ZJU100 Young Professor with both the College of Computer Science and the School of Cyber Science and Technology at Zhejiang University, China. He was an assistant professor at the Department of Electrical and Computer Engineering in the University of Texas at San Antonio after receiving the Ph.D. degree from the Computer Science and Engineering department at State University of New York at Buffalo in 2017. His current research interests include data security and privacy, secure computation outsourcing, artificial intelligence security, and cyber-physical security in the context of the Internet of Things. His works explore and develop novel security sensitive algorithms and protocols for computation and communication on the general context of Cloud and Internet devices. He is the associate editor of IEEE TDSC.
\end{IEEEbiography}

\end{document}